\documentclass[11pt,letterpaper]{amsart}
\pdfoutput=1

\usepackage{amssymb}
\usepackage{mathrsfs}
\usepackage{euscript}
\usepackage{graphicx}
\usepackage{rotating}
\newcommand{\myV}[1]{V({#1})}
\newcommand{\myE}[1]{E({#1})}
\newcommand{\myc}[1]{c({#1})}
\newcommand{\myI}[2]{{#1}[{#2}]}
\newcommand{\mygG}{G}
\newcommand{\mygH}{H}
\newcommand{\mygC}{C}

\newcommand{\myieq}[1]{{\equiv}_{#1}\,}
\newcommand{\myvex}{\exp\bigl(O(n)\bigr)}
\newcommand{\mydeln}{\backslash}

\newenvironment{comm}{\begin{list}{{[[}}{\setlength{\leftmargin}{3.75mm}}\item}{{]]}\end{list}\smallskip}

\newtheorem{Thm}{Theorem}
\newtheorem{Lem}[Thm]{Lemma}
\newtheorem{Cor}[Thm]{Corollary}
\theoremstyle{definition}
\newtheorem{Def}[Thm]{Definition}

\newenvironment{Proof}{\begin{proof}}{\end{proof}}

\allowdisplaybreaks

\begin{document}


\hyphenation{University}

\title%
[]%
{Computing the Tutte polynomial\\ in vertex-exponential time}
\author[]{Andreas Bj\"{o}rklund$^*$}
\thanks{$^*$Lund University,
Department of Computer Science,
P.O.Box 118, SE-22100 Lund, Sweden. E-mail: {\tt
andreas.bjorklund@logipard.com}, {\tt thore.husfeldt@cs.lu.se}.}  
\author[]{Thore Husfeldt$^*$}
\author[]{Petteri Kaski$^\dagger$}
\thanks{$^\dagger$Helsinki Institute for Information Technology HIIT,
Department of Computer Science, University of Helsinki,
P.O.Box 68, FI-00014 University of Helsinki, Finland. 
E-mail: {\tt petteri.kaski@cs.helsinki.fi}, {\tt
mikko.koivisto@cs.helsinki.fi}. 
This research was supported in part by the Academy of
Finland, Grants 117499 (P.K.) and 109101 (M.K.)}
\author[]{Mikko Koivisto$^\dagger$}


\begin{abstract}
The deletion--contraction algorithm is perhaps
the most popular method for computing a host of fundamental
graph invariants such as the chromatic, flow, and reliability
polynomials in graph theory,
the Jones polynomial of an alternating link in knot theory,
and the partition functions of the models of Ising, Potts, and
Fortuin--Kasteleyn in statistical physics.
Prior to this work, deletion--contraction  was also the fastest known
general-purpose algorithm for these invariants, running in time
roughly proportional to the number of spanning trees in the input graph.

Here, we give a substantially faster algorithm that computes the Tutte
polynomial---and hence, all the aforementioned invariants and more---of an
arbitrary graph in time within a polynomial factor of the number of connected
vertex sets. The algorithm actually evaluates a multivariate generalization of
the Tutte polynomial by making use of an identity due to Fortuin and Kasteleyn.
We also provide a polynomial-space variant of the algorithm and
give an analogous result for Chung and Graham's cover polynomial.

An implementation of the algorithm outperforms
deletion--contraction also in practice.
\end{abstract}


\maketitle

\vspace*{-22pt}

\section{Introduction}

Tutte's motivation for studying what he called the 
``dichromatic polynomial'' was algorithmic. By his own entertaining 
account \cite{Tut04}, he was intrigued by the variety of graph 
invariants that could be computed with the deletion--contraction 
algorithm, and ``playing'' with it he discovered a bivariate 
polynomial that we can define as
\begin{equation}
\label{eq:tutte}
T_\mygG(x,y)= \sum_{F\subseteq E} (x-1)^{\myc F-\myc E}
(y-1)^{\myc F+|F|-|V|}\,.
\end{equation}
Here, $\mygG$ is a graph with vertex set $V$ and edge set $E$; 
by $\myc F$ we denote the number of connected 
components in the graph with vertex set $V$ and edge set $F$.
Later, Oxley and Welsh \cite{OW79} showed in their
celebrated Recipe Theorem that, in a very 
strong sense, the {\em Tutte polynomial} $T_\mygG$ is indeed the most 
general graph invariant that can be computed using deletion--contraction.

Since the 1980s it has become clear that this construction has deep
connections to many fields outside of computer science and algebraic
graph theory. It appears in various guises and specialisations in
enumerative combinatorics, statistical physics, knot theory and
network theory. It subsumes the chromatic, flow, and reliability
polynomials, the Jones polynomial of an alternating link, and, perhaps
most importantly, the models of Ising, Potts, and Fortuin--Kasteleyn,
which appear in tens of thousands of research papers. A number of
surveys written for various audiences present and explain these
specialisations \cite{Soka05, Welsh93, Wel99, WM99}.

Computing the Tutte polynomial has been a very fruitful
topic in theoretical computer science, resulting in seminal work on
the computational complexity of counting, several algorithmic
breakthroughs both classical and quantum, and whole research
programmes devoted to the existence and nonexistence of approximation
algorithms. Its specialisation to graph colouring has been one of the 
main benchmarks of progress in exact algorithms.

The deletion--contraction algorithm computes $T_\mygG$ for a connected
$\mygG$ in time within a polynomial factor of $\tau(\mygG)$, the
number of spanning trees of the graph, and no essentially faster
algorithm was known. In this paper we show that the Tutte
polynomial---and hence, by virtue of the Recipe Theorem, \emph{every}
graph invariant admitting a deletion--contraction recursion---can be
computed in time within a polynomial factor of $\sigma(\mygG)$, the
number of vertex subsets that induce a connected subgraph.
Especially, the algorithm runs in time $\exp\bigl(O(n)\bigr)$, that
is, in ``vertex-exponential'' time, while $\tau(G)$ typically is
$\exp\bigl(\omega(n)\bigr)$ and can be as large as
$n^{n-2}$~\cite{Cayl89}.  Previously, vertex-exponential running time
bounds were known only for evaluations of $T_\mygG$ in special regions
of the Tutte plane $(x,y)$, such as for the chromatic polynomial and
(using exponential space) the reliability polynomial,
or only for special classes of graphs such as planar graphs or
bounded-degree graphs. We provide a more detailed overview of such
prior work in \S\ref{sect:prior}.

\begin{figure}
\includegraphics{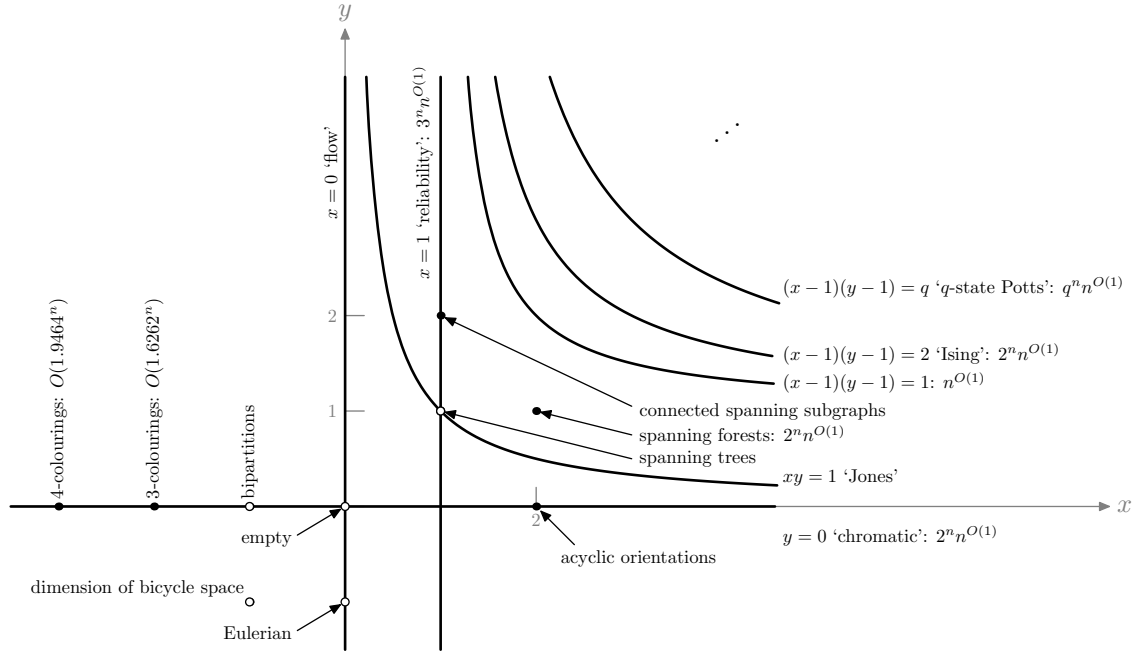}
\small
\caption{\label{fig: atlas} An atlas of the Tutte plane $(x,y)$.  The
  five points shown by circles and the points on the hyperbola
  $(x-1)(y-1)=1$ are in P, all other points are \#P-complete.  Those
  points and lines where algorithms with complexity $\myvex$ were
  previously known (sometimes only in exponential space), are labelled
  with their running time; note that the hyperbolas $(x-1)(y-1)=q$
  were known to be vertex-exponential only for fixed integer $q$. See
  \S\ref{sec: regions} for references.  Our result is that the entire
  plane admits algorithms with running time $2^nn^{O(1)}$ and
  exponential space, or time $3^nn^{O(1)}$ and polynomial space. The
  only points that are known to admit algorithms with better bounds are the
  ``colouring'' points $(-2,0)$ and $(-3,0)$, the ``Ising'' hyperbola
  $(x-1)(y-1)=2$, for which a faster algorithm in observed in
  \S\ref{sec: regions}, and of course the points in P.  (Only the
  positive branches of the hyperbolas are drawn.)}
\vspace*{-0.3mm}
\end{figure}

\subsection{Result and consequences}

By ``computing the Tutte polynomial'' we mean computing the
coefficients $t_{ij}$ of the monomials $x^iy^j$ in $T_\mygG(x,y)$ 
for a graph $\mygG$ given as input. Of course, the coefficients 
also enable the efficient evaluation of $T_\mygG(x,y)$ at any given 
point $(x,y)$. Our main result is as follows.

\begin{Thm}
\label{thm:connected-main}
  The Tutte polynomial of an $n$-vertex graph $\mygG$ 
  can be computed
\begin{enumerate}
\item[(a)] in time and space $\sigma(\mygG) n^{O(1)}$;
\item[(b)] in time $3^nn^{O(1)}$ and polynomial space; and
\item[(c)] in time $3^{n-s}2^{s}n^{O(1)}$
           and space $2^{s}n^{O(1)}$ for any integer $s$,\, $0 \leq s \leq n$.
\end{enumerate}
\end{Thm}
\vspace*{-0.1mm}

Especially, the Tutte polynomial can be evaluated everywhere in
vertex-exponential time. In some sense, this is both surprising and
optimal, a claim that we solidify under the Exponential Time
Hypothesis in \S\ref{sec: Complexity}.  Moreover, even for those
curves and points of the Tutte plane where a vertex-exponential time
algorithm was known before, our algorithm improves or at least
matches their performance, with only a few exceptions (see
Figure~\ref{fig: atlas}).

For bounded-degree graphs $G$, the deletion--contraction algorithm itself runs
in vertex-exponential time because $\tau(G)=\exp\bigl(O(n)\bigr)$. Theorem
\ref{thm:connected-main}  still gives a better bound because it is known that
$\sigma(G)= O((2-\epsilon)^n)$ for bounded degree \cite[Lemma 6]{BHKK08}, while
$\tau(G)$ grows faster than $2.3^n$ already for 3-regular graphs (see
\S\ref{sec: restricted classes}).  The precise bound is as follows:

\vspace*{-0.1mm}
\begin{Cor}
\label{thm:bounded-degree}
  The Tutte polynomial of an $n$-vertex graph with maximum
  vertex degree $\Delta$ can be computed in time 
  $\xi_{\Delta}^nn^{O(1)}$, where 
  $\xi_\Delta=(2^{\Delta+1}-1)^{1/(\Delta+1)}$.
\end{Cor}
\vspace*{-0.1mm}

The question about solving deletion--contraction based algorithmic
problems in vertex-exponential time makes sense in \emph{directed}
graphs as well.  Here, the most successful attempt to define an
analogue of the Tutte polynomial is Chung and Graham's \emph{cover
  polynomial}, which satisfies directed analogues to the
deletion--contraction operations \cite{CG95}. It turns out that a
directed variant of our main theorem can be established using recent
techniques that are by now well understood, we include the precise
statement and proof in Appendix~\ref{sec: cover}.

\subsection{Overview of techniques}

The Tutte polynomial is, in essence, a sum over connected spanning
subgraphs.  Managing this connectedness property introduces a
computational challenge not present with its specialisations, e.g., with
the chromatic polynomial. Neither the dynamic programming algorithm
across vertex subsets by Lawler \cite{Law76} nor the recent
inclusion--exclusion algorithm \cite{BHK07}, which apply for counting
$k$-colourings, seems to work directly for the Tutte
polynomial. Perhaps surprisingly, they do work for the cover
polynomial, even though the application is quite involved; the details
are in Appendix \ref{sec: cover} and can be seen as an attempt to
explain just how far these concepts get us.

For the Tutte polynomial, we take a detour via the Potts model.  The
idea is to evaluate the partition function of the $q$-states Potts
model at suitable points using inclusion--exclusion, which then, by a
neat identity due to Fortuin and Kasteleyn \cite{FK72,Soka05}, enables the
evaluation of the Tutte polynomial at any given point by polynomial
interpolation. Finally, another round of polynomial interpolation
yields the desired coefficients of the Tutte polynomial. Each step can
be implemented using only polynomial space.  Moreover, the approach
readily extends to the multivariate Tutte polynomial of Sokal
\cite{Soka05} which allows the incorporation of arbitrary edge
weights; that generalisation can be communicated quite concisely using
the involved high-level framework, which we do in
\S\ref{section:tutte-potts}. To finally arrive at the main result of
this paper---reducing the running time to within a polynomial
factor of $\sigma(G)$---requires
manipulation at the level of the fast Moebius transform ``inside'' the
algorithm, which can be found in \S\ref{section:connected}. The smooth
time--space tradeoff, Theorem~\ref{thm:connected-main}(c), is
obtained by a new ``split transform'' technique
(Appendix~\ref{appendix:split-trans}).

Our approach highlights the algorithmic significance of the Fortuin--Kasteleyn
identity, and suggests a more general technique: to compute a polynomial, it may
be advisable to look at its evaluations at integral (or otherwise special) points,
with the objective of obtaining new combinatorial or algebraic interpretations
that then enable faster reconstruction of the entire polynomial. (For example,
the multiplication of polynomials via the fast Fourier transform can be
seen as an instantiation of this technique.) 

We also give another vertex-exponential time algorithm that does not
rely on interpolation (\S\ref{section:recurrence}). It is based on a
new recurrence formula that alternates between partitioning an induced
subgraph into components and a subtraction step to solve the connected
case. The recurrence can be solved using fast subset convolution
\cite{BHKK07a} over a multivariate polynomial ring. However, an
exponential space requirement seems inherent to that
algorithm.\footnote{A previous version of this manuscript followed
  this route, establishing Theorem 1(a).}
Appendix~\ref{appendix:implementation} briefly reports on our
experiences with implementing and running this algorithm; it
outperforms deletion--contraction in the worst case when $n\geq 13$.

\subsection{Conventions}

For standard graph-theoretic terminology we refer to West
\cite{West01}.  All graphs we consider are undirected and may contain
multiple edges and loops. For a graph $\mygG$, we write $n=n(\mygG)$
for the number of vertices, $m=m(\mygG)$ for the number of edges,
$V=\myV\mygG$ for the vertex set, $E=\myE\mygG$ for the edge set,
$c=\myc\mygG$ for the number of connected components, 
$\tau(G)$ for the number of spanning trees,
and $\sigma(G)$ the {\em number of connected sets}, i.e., the number
of vertex subsets that induce a connected graph.

To simplify running time bounds, we assume $m=n^{O(1)}$ and remark 
that this assumption is implicit already in Theorem 
\ref{thm:connected-main}.  (Without this assumption, all the time 
bounds require an additional multiplicative term $m^{O(1)}$.)  
For a set of vertices $U\subseteq\myV\mygG$, we
write $\myI\mygG U$ for the subgraph induced by $U$ in $\mygG$.  
A subgraph $\mygH$ of $\mygG$ is \emph{spanning} if
$\myV\mygH=\myV\mygG$.  For a proposition $P$, we use Iverson's
bracket notation $[P]$ to mean $1$ if $P$ is true and $0$ otherwise.

\section{Prior work: Algorithms for the Tutte Polynomial}
\label{sect:prior}

The direct evaluation of $T_\mygG(x,y)$ based on \eqref{eq:tutte}
takes $2^mn^{O(1)}$ steps and polynomial space, but many other
expansions have been studied in the literature.

\subsection{Spanning Tree Expansion}
\label{sect:spanning-tree-expansion}

If we expand and collect terms in \eqref{eq:tutte} we arrive at
\begin{equation}
\label{eq:spanning-tree-expansion}
T_\mygG(x,y)= \sum_{i,j} t_{ij} x^i y^j\,.
\end{equation}
In fact, this is Tutte's original definition.  The coefficients
$t_{ij}$ of this expansion are well-studied: 
assuming that $\mygG$ is connected, 
$t_{ij}$ is the number of spanning trees of $\mygG$ having
``internal activity'' $i$ and ``external activity'' $j$.  What these
concepts mean need not occupy us here (for example, see
\cite[\S13]{Bigg93}), for our purposes it is sufficient to know that
they can be efficiently computed for a given spanning tree. Thus
\eqref{eq:spanning-tree-expansion} can be evaluated directly by iterating
over all spanning trees of $\mygG$, which can be accomplished with
polynomial delay \cite{KR95}. The resulting running time is within a
polynomial factor of $\tau(\mygG)$.

Some of the coefficients $t_{ij}$ have an alternative 
combinatorial interpretation, and some can be computed faster 
than others. For example, $t_{00}=0$ holds if $m>0$,
and $t_{01}=t_{10}$ if $m>1$. The latter value, 
the \emph{chromatic invariant} $\theta(\mygG)$, 
can be computed from the chromatic polynomial, 
and thus can be found in time $2^nn^{O(1)}$ \cite{BHK07}.

The computational complexity of computing individual coefficients 
$t_{ij}$ has also been investigated. In particular, polynomial-time 
algorithms exist for $t_{n-1-k,j}$ for constant $k$ and all 
$j=0,1,\ldots,m-n+1$. In general, the task of computing 
$t_{ij}$ is \#P-complete \cite{Ann92}.

\subsection{Deletion--Contraction}
\label{sect:deletion--contraction}

The classical algorithm for computing $T_\mygG$ is the
following \emph{deletion--contraction} algorithm. It is
based on two graph transformations involving an edge $e$.
The graph $\mygG\mydeln e$ is obtained from $G$ by \emph{deleting} $e$.
The graph $\mygG/e$ is obtained from $G$ by \emph{contracting} $e$,
that is, by identifying the endvertices of $e$ and then deleting $e$.

With these operations, one can establish the recurrence formula
\begin{equation}\label{eq: deletion-contraction}
  T_\mygG(x,y) =\begin{cases}
    1                                     & \text{if $\mygG$ has no edges;} \\
    y T_{\mygG\mydeln e}(x,y)                & \text{if $e$ is a loop;} \\
    x T_{\mygG/e}(x,y)                       & \text{if $e$ is a bridge;} \\
    T_{\mygG\mydeln e}(x,y)+T_{\mygG/e}(x,y) & \text{otherwise.}\\
\end{cases}
\end{equation}

The deletion--contraction algorithm defined by a direct evaluation of
\eqref{eq: deletion-contraction} leads to 
a running time that scales as the Fibonacci sequence, 
$\bigl((1+\sqrt 5)/2\bigr)^{n+m}=O(1.6180^{n+m})$ \cite{Wilf}.
Sekine, Imai, and Tani \cite{SIT95} observed that the
corresponding computation tree has one leaf for every spanning tree of
$\mygG$, so \eqref{eq: deletion-contraction} is yet another way to
evaluate $T_G$ in time within a polynomial factor of $\tau(\mygG)$. In
practice one can speed up the computation by identifying isomorphic
graphs and using dynamic programming to avoid redundant recomputation
\cite{HPR,Imai00,SIT95}.

The deletion--contraction algorithm is known to compute many different
graph parameters. For example, the number of spanning trees admits an
analogous recursion, as does the number of acyclic orientations, the
number of colourings, the dimension of the bicycle space, and so forth
\cite[\S15.6--8]{GR01}. This is no surprise: all these graph
parameters are evaluations of the Tutte polynomial at certain points.
But not only is every specialisation of $T_\mygG$ expressible by
deletion--contraction, the converse holds as well: \emph{every} graph
parameter that can be expressed as a deletion--contraction recursion
turns out to be a valuation of $T_\mygG$, according to the celebrated
Recipe Theorem of Oxley and Welsh \cite{OW79}
(cf.~\cite[Theorem~X.2]{Boll98}).

Besides deletion--contraction, many other expansions are known 
(in particular for restrictions of the Tutte polynomial; see \cite{Bigg93}), 
even a convolution over the set of edges \cite{KRS99}, 
but none leads to vertex-exponential time.

\subsection{Regions of the Tutte plane}
\label{sec: regions}

The question at which points $(x,y)$ the Tutte polynomial can be
computed exactly and efficiently was completely settled in the
framework of computational complexity in the seminal paper of Jeager,
Vertigan, and Welsh \cite{JVW90}: They presented a complete classification 
of points and curves where the problem is polynomial-time computable, and
where it is \#P-complete. This result shows us where we probably 
need to resign ourselves to a superpolynomial-time algorithm.

For most of the \#P-hard points, the algorithms from
\S\ref{sect:spanning-tree-expansion} and
\S\ref{sect:deletion--contraction} were best known.  However, for
certain regions of the Tutte plane, algorithms running in time
$\myvex$ have been known before. We attempt to summarise these
algorithms here, including the polynomial-time cases; see
Figure~\ref{fig: atlas}.

\begin{description}
\item[Trivial hyperbola] On the hyperbola $(x-1)(y-1)= 1$ 
    the terms of \eqref{eq:tutte} involving $\myc F$ cancel, so
  $T_\mygG(x,y)=(x-1)^{n-c}y^m$, which can be evaluated
  in polynomial time.
\item[Ising model] On the hyperbola $H_2 \equiv (x-1)(y-1) = 2$, 
  the Tutte polynomial gives the partition function of the 
  \emph{Ising model}, a sum of easily computable weights over the 
  $2^n$ configurations of $n$ two-state spins. This can be
  trivially computed in time $2^n n^{O(1)}$ and polynomial
  space.  
  By dividing the $n$ spins into three groups of
  about equal size and using fast matrix multiplication, one
  can compute the sum in time $2^{n \omega/3} n^{O(1)} = O(1.732^n)$
  and exponential space, where $\omega$ is the exponent of
  matrix multiplication; this is yet a new  
  application of Williams's trick \cite{BH07,Koi06,Wil04}.
\item[Potts model] More generally, for any integer $q \geq 2$, the
  Tutte polynomial on the hyperbola $H_q \equiv (x-1)(y-1) = q$ gives
  the partition function of the \emph{q-state Potts model}
  \cite{Potts_1952}.  This is a sum over the configurations of $n$
  spins each having $q$ possible states. It can be computed trivially
  in time $q^nn^{O(1)}$ and, via fast matrix multiplication, in time
  $q^{n3/\omega}n^{O(1)}$.  We will show in
  \S\ref{section:tutte-potts} that, in fact, time $2^nn^{O(1)}$
  suffices, which result will be an essential building block in our
  main construction.
\item[Reliability polynomial] The reliability polynomial
  $R_G(p)$, which is the probability that no component of $G$ is
  disconnected after independently removing each edge with 
  probability $1-p$, satisfies $R_G(p) = p^{m-n+c}(1-p)^{n-c} T_G(1, 1/p)$ 
  and can be evaluated in time $3^{n} n^{O(1)}$ 
  and exponential space \cite{Buza80}.
\item[Number of spanning trees] For connected $\mygG$, $T_\mygG(1,1)$
  equals the number $\tau(\mygG)$ of spanning trees, and is
  computable in polynomial time as the determinant of 
  a maximal principal submatrix of the Laplacian of $\mygG$, 
  a result known as Kirchhoff's Matrix--Tree Theorem.
\item[Number of spanning forests] 
  The number of spanning forests, $T_\mygG(2,1)$, is
  computable in time $2^nn^{O(1)}$ by first using the Matrix--Tree Theorem
  for each induced subgraph and then assembling the result 
  one component (that is, tree) at a time via
  inclusion--exclusion 
  \cite{BHK07}. (This observation is new to the present 
  work, however.)
\item[Dimension of the bicycle space] $T_\mygG(-1,-1)$ computes the
  dimension of the bicycle space, in polynomial time by Gaussian
  elimination.
\item[Number of nowhere-zero 2-flows] $T_\mygG(0,-1)=1$ if $G$ is
  Eulerian (in other words, it ``admits a nowhere-zero 2-flow''), and
  $T_\mygG(0,-1)=0$ otherwise. Thus $T_\mygG(0,-1)$ is computable
  in polynomial time.
\item[Chromatic polynomial] The chromatic polynomial $P_\mygG(t)$,
  which counts the number of proper $t$-colourings of the vertices of
  $\mygG$, satisfies $P_\mygG(t)=(-1)^{n-c}t^cT_\mygG(1-t,0)$ and can
  be computed in time $2^nn^{O(1)}$ \cite{BHK07}. Vertex-exponential
  time algorithms were known at least since Lawler~\cite{Law76}, and a
  vertex-exponential, polynomial-space algorithm was found only
  recently \cite{BH07}.  Other approaches to the chromatic polynomial
  are surveyed by Anthony \cite{Ant90}.  At $t=2$ (equivalently,
  $x=-1$) this is polynomial-time computable by breadth-first search
  (every connected component of a bipartite graph has exactly two
  proper 2-colourings). The cases $t=3,4$ are well-studied benchmarks
  for exact counting algorithms, the current best bounds are
  $O(1.6262^n)$ and $O(1.9464^n)$ \cite{FGS07}.  The case $x=0$ is
  trivial.
\end{description}

To the best knowledge of the authors, no algorithms with running time
$\myvex$ have been known for other real points. If we allow
$x$ and $y$ to be complex, there are four more points $(x,y)$
at which $T_\mygG$ can be evaluated in polynomial time \cite{JVW90}.

\subsection{Restricted graph classes}
\label{sec: restricted classes}
Explicit formulas for Tutte polynomial have been derived
for many elementary families of graphs, such as 
$T(\mygC_n;x,y)= y+x+x^2+\cdots+x^{n-1}$ for
the $n$-cycle graph $\mygC_n$. We will not give an overview of these 
formulas here (see \cite[\S13]{Bigg93}); most of 
them are applications of deletion--contraction.

For well-known graph classes, the authors know the following results 
achieving $\myvex$ running time or better:

\begin{description}
\item[Planar graphs] If $\mygG$ is planar, then the Tutte polynomial can be
  computed in time $\exp\bigl(O(\sqrt{n}\,)\bigr)$ \cite{SIT95}. This
  works more generally, with a slight overhead: in classes of graphs
  with separators of size $n^{\alpha}$, the Tutte polynomial can
  be computed in time $\exp\bigl(O(n^{\alpha}\log n)\bigr)$.

\item[Bounded tree-width and branch-width] For $k$ a fixed integer, if
  $\mygG$ has tree-width $k$ then $T_\mygG$ can be computed in
  polynomial time \cite{Andr98, Nob98}. This can be generalised to
  branch-width \cite{Hli06}.

\item[Bounded clique-width and cographs] For $k$ a fixed integer, if
  $\mygG$ has clique-width $k$ then $T_\mygG$ can be computed in time
  $\exp\bigl(O(n^{1-1/(k+2)})\bigr)$ \cite{GHN06}. A special case of this 
  is the class of cographs (graphs without an induced path of 4 vertices), 
  where the bound becomes $\exp\bigl(O(n^{2/3})\bigr)$.

\item[Bounded-degree graphs] If $\Delta$ is the maximum degree of a
  vertex, the deletion--contraction algorithm and $2m\leq n\Delta$
  yield the vertex-exponential running time bound
  $O\bigl(1.6180^{(1+\Delta/2)n}\bigr)$ directly from the recurrence.
  Gebauer and Okamoto improve this to $\chi_\Delta^nn^{O(1)}$, where
  $\chi_\Delta=2(1-\Delta 2^{-\Delta})^{1/(\Delta+1)}$ (for example,
  $\chi_3= 2.5149$, $\chi_4=3.7764$, and $\chi_5=5.4989$).  For
  $k$-regular graphs with $k\geq 3$ a constant independent of $n$, the
  number of spanning trees (and hence, within a polynomial factor, the
  running time of the deletion--contraction algorithms) is bounded by
  $ \tau(G)=O\bigl(\nu_k^nn^{-1}\log n\bigr)$, where
  $\nu_k=(k-1)^{k-1}/(k^2-2k)^{k/2-1}$ (for example, $\nu_3=2.3094$,
  $\nu_4=3.375$, and $\nu_5=4.4066$), and this bound is tight \cite{ChYa99}.
  
\item[Interval graphs]\sloppy If $\mygG$ is an interval graph, then $T_\mygG$
  can be computed in time $O(1.9706^m)$, which is not $\myvex$ in
  general, but still faster than by deletion--contraction
  \cite{GO07}.
\end{description}

What we cannot survey here is the extensive literature that 
studies algorithms that simultaneously specialise  $T_\mygG$ and restrict
the graph classes, often with the goal of developing a polynomial-time
algorithm. A famous example is that for Pfaffian orientable graphs,
which includes the class of planar graphs, the Tutte polynomial is
polynomial-time computable on the hyperbola $H_2$ \cite{Kast61}. 
Within computer science, the most studied specialisation of
this type is most likely graph colouring for restricted graph classes.

\subsection{Computional complexity}
\label{sec: Complexity}

The study of the computational complexity of the Tutte polynomial
begins with Valiant's theory of \#P-completeness \cite{Val79} and the
exact complexity results of Jaeger, Vertigan, and Welsh \cite{JVW90}.
The study of the approximability of the values of $T_\mygG$ has been a
very fruitful research direction, an overview of which is again
outside the scope of this paper. In this regard we refer to Welsh's
monograph \cite{Welsh93} and to the recent paper of Goldberg and
Jerrum \cite{GJ07} for a survey of newer developments.

For our purposes, the most relevant hardness results have been
established under the Exponential Time Hypothesis \cite{IPZ01}
(ETH). First, deciding whether a given graph can be 3-coloured requires
$\exp(\Omega(n))$ time under ETH, and since 3-colourability can be decided by
computing $T_\mygG(-2,0)$ we see that evaluating the Tutte polynomial
requires vertex-exponential time under ETH. Thus, it would be
surprising if our results could be significantly improved, for example
to something like $\exp\big(O(n/\log n)\big)$.

Second, it is by no means clear that the entire Tutte plane should
admit such algorithms. Many specialisations of the Tutte polynomial
can be understood as constraint satisfaction problems. For example,
graph colouring is an instance of $(q,2)$-CSP, the class of constraint
satisfaction problems with pairwise constraints over $q$-state
variables. Similarly, the partition function for the Potts model can
be seen as a weighted counting CSP \cite{DGJ07}. Very recently,
Traxler \cite{Tra08} has shown that already the decision version of
$(q, 2)$-CSP requires time $\exp\big(\Omega(n\log q)\big)$ under ETH,
even for some very innocent-looking restrictions, and even for bounded
degree graphs. Thus in general, these CSPs are not vertex-exponential
under ETH.

\section{The multivariate Tutte polynomial via 
         the $q$-state Potts model}
\label{section:tutte-potts}

Let $R$ be a multivariate polynomial ring over a field 
and let $G$ be an undirected graph with vertex 
set $V=\{1,2,\ldots,n\}$ and edge set $E$, $m=n^{O(1)}$.
We allow $G$ to have parallel edges and loops.
Associate with each $e\in E$ a ring element $r_e\in R$.
The {\em multivariate Tutte polynomial} \cite{Soka05} of $G$ is the
polynomial
\begin{equation}
\label{eq:multivariate-tutte}
Z_G(q,r)=\sum_{F\subseteq E}q^{c(F)}\prod_{e\in F}r_e\,,
\end{equation}
where $q$ is an indeterminate and $c(F)$ denotes the number 
of connected components in the graph with vertex set $V$ 
and edge set $F$. The product over an empty set always 
evaluates to $1$.

The classical Tutte polynomial $T_G(x,y)$ can be recovered 
as a bivariate evaluation of the multivariate polynomial 
$Z_G(q,r)$ via
\begin{equation}
\label{eq:multivariate-eval-classical}
T_G(x,y)=(x-1)^{-c(E)}(y-1)^{-|V|}Z_G\bigl((x-1)(y-1),y-1\bigr)\,.
\end{equation}

\subsection{The Fortuin--Kasteleyn identity}

At points $q=1,2,\ldots$ the multivariate Tutte polynomial 
$Z_G(q,r)$ can be represented as an evaluation of the 
partition function of the $q$-state Potts model \cite{FK72,Soka05}.

For a mapping $s:V\rightarrow \{1,2,\ldots,q\}$ and 
an edge $e\in E$ with endvertices $x$ and $y$, 
define $\delta_e^s=1$ if $s(x)=s(y)$ and $\delta_e^s=0$ if $s(x)\neq s(y)$.
The {\em partition function} of the $q$-state Potts model
on $G$ is defined by
\begin{equation}
\label{eq:potts}
Z_G^{\text{Potts}}(q,r)=
\sum_{s:V\rightarrow \{1,2,\ldots,q\}}
  \prod_{e\in E}\bigl(1+r_e\delta_e^s\bigr)\,.
\end{equation}

\begin{Thm}[Fortuin and Kasteleyn]
\label{thm:fortuin--kasteleyn}
For all $q=1,2,\ldots$ it holds that 
\begin{equation}
\label{eq:fortuin--kasteleyn}
Z_G(q,r) = Z_G^{\text{Potts}}(q,r)\,.
\end{equation}
\end{Thm}

\subsection{The multivariate Tutte polynomial via the $q$-state Potts model}

By virtue of the Fortuin--Kasteleyn identity 
\eqref{eq:fortuin--kasteleyn},
to compute $Z_G(q,r)$ it suffices to evaluate
\[
 Z_G^{\text{Potts}}(1,r),\ 
 Z_G^{\text{Potts}}(2,r),\ 
 \ldots,\ 
 Z_G^{\text{Potts}}(n+1,r)
\]
and then recover $Z_G(q,r)$ via Lagrangian interpolation.
For the interpolation to succeed, it is necessary to assume that 
the coefficient field of $R$ has a large enough characteristic 
so that $1,2,\ldots,n$ have multiplicative inverses.

At first sight the evaluation of \eqref{eq:potts}
for a positive integer $q$ appears to require $q^nn^{O(1)}$ ring 
operations. Fortunately, one can do better.
To this end, let us express $Z_G^{\text{Potts}}(q,r)$ in 
a more convenient form.
For $X\subseteq V$, denote by $G[X]$ the subgraph of $G$ 
induced by $X$, and let 
\begin{equation}
\label{eq:potts-f}
f(X)=\prod_{e\in E(G[X])}(1+r_e)\,.
\end{equation}
For $q=1,2,\ldots$, we have
\begin{equation}
\label{eq:potts-partition}
Z_G^{\text{Potts}}(q,r)\ =\!\!\!\sum_{(U_1,U_2,\ldots,U_q)}\!\!\!
f(U_1)f(U_2)\cdots f(U_q)\,,
\end{equation}
where the sum is over all $q$-tuples $(U_1,U_2,\ldots,U_q)$
with $U_1,U_2,\ldots,U_q\subseteq V$ such 
that $\cup_{i=1}^q U_i=V$ and 
$U_j\cap U_k\neq\emptyset$ for all $1\leq j<k\leq q$.

We now proceed to develop algorithms for evaluating
the Potts partition function in the form \eqref{eq:potts-partition}.

\subsection{The baseline algorithm}

Let $f:2^V\rightarrow R$ be a function that associates
a ring element $f(X)\in R$ with each subset $X\subseteq V$.

The {\em zeta transform} $f\zeta:2^V\rightarrow R$ is defined 
for all $Y\subseteq V$ by $f\zeta(Y)=\sum_{X\subseteq Y}f(X)$.
The {\em Moebius transform} $f\mu:2^V\rightarrow R$ is defined 
for all $X\subseteq V$ by 
$f\mu(X)=\sum_{Y\subseteq X}(-1)^{|X\setminus Y|}f(Y)$.

It is a basic fact that the zeta and Moebius transforms
are inverses of each other, that is, $f\zeta\mu=f\mu\zeta=f$
for all $f$. Furthermore, it is known \cite{BHKK07a} that
\begin{equation}
\label{eq:fcovering}
\bigl((f\zeta)^q\mu\bigr)(V)=
\sum_{(U_1,U_2,\ldots,U_q)}f(U_1)f(U_2)\cdots f(U_q)\,,
\end{equation}
where the sum is over all $q$-tuples $(U_1,U_2,\ldots,U_q)$ with
$U_1,U_2,\ldots,U_q\subseteq V$ and $\cup_{j=1}^q U_j=V$.
In particular, $\bigl((f\zeta)^q\mu\bigr)(V)$ can be computed
directly in $3^nn^{O(1)}$ ring operations by storing $n^{O(1)}$ 
ring elements. Using the fast zeta and Moebius transforms, 
$\bigl((f\zeta)^q\mu\bigr)(V)$ can be computed
in $2^nn^{O(1)}$ ring operations by storing $2^nn^{O(1)}$ ring elements
\cite{BHKK07a}.

To use this to evaluate \eqref{eq:potts-partition}, adjoin 
a new indeterminate $z$ into $R$ to obtain the polynomial ring $R[z]$.
Replace $f$ with $f_z:2^V\rightarrow R[z]$ defined 
for all $X\subseteq V$ by $f_z(X)=f(X)z^{|X|}$. 
Now evaluate the $z$-polynomial $\bigl((f_z\zeta)^q\mu\bigr)(V)$ 
and look at the coefficient of the monomial $z^{|V|}$,
which by virtue of \eqref{eq:fcovering} is equal 
to \eqref{eq:potts-partition}.

This baseline algorithm together with \eqref{eq:multivariate-eval-classical}, 
\eqref{eq:fortuin--kasteleyn}, and Lagrangian interpolation 
establishes that the Tutte polynomial $T_G(x,y)$ can be computed 
(a) in time and space $2^nn^{O(1)}$; and
(b) in time $3^nn^{O(1)}$ and space $n^{O(1)}$.
This proves Theorem \ref{thm:connected-main}(b).
A more careful analysis of $\bigl((f\zeta)^q\mu\bigr)(V)$ 
enables the time--space tradeoff in Theorem~\ref{thm:connected-main}(c).
[[ See Appendix \ref{appendix:split-trans}. ]]

\section{Improvements and variations}

\subsection{An algorithm over connected sets}
\label{section:connected}

It is useful to think of $X\subseteq V$ in what 
follows as the current subset under consideration.
We start with a lemma that partitions the subsets of $X$ 
based on the maximum common suffix. 
To this end, let $Y\myieq i X$ be a shorthand 
for $Y\cap\{i+1,i+2,\ldots,n\}=X\cap\{i+1,i+2,\ldots,n\}$.
\begin{Lem}[Suffix partition]
\label{lem:suffix-partition}
Let $Y\subseteq X\subseteq\{1,2,\ldots,n\}$.
Then, either $Y=X$ or there exists a unique $i\in X$
such that $Y\myieq{i-1} X\setminus\{i\}$.
\end{Lem}
\begin{Proof}
Either $Y=X$ or $i=\max X\setminus Y$.
\end{Proof}

The intermediate values computed by the algorithm are now 
defined as follows.

\begin{Def}
\label{def:fpoly}
Let $X\subseteq V$, $q=1,2,\ldots,n+1$, and $i=0,1,\ldots,n$.
Let 
\[
F(X,q,i)=\sum_{(U_1,U_2,\ldots,U_q)}\prod_{j=1}^q f(U_j)\,,
\]
where the sum is over all $q$-tuples $(U_1,U_2,\ldots,U_q)$
such that both $U_1,U_2,\ldots,U_q\subseteq X$ and
$\cup_{j=1}^q U_j\myieq{i} X$.
\end{Def}

Note that $F(V,q,0)=((f\zeta)^q\mu)(V)$. Thus, it suffices to
compute $F(V,q,0)$.

We are now ready to describe the algorithm that computes the
intermediate values $F(X,q,i)$ in Definition \ref{def:fpoly}.
The algorithm considers one set $X\subseteq V$ at a time, starting 
with the empty set $X=\emptyset$ and proceeding upwards in the subset 
lattice. It is required that the maximal proper subsets of $X$
have been considered before $X$ itself is considered; for example, 
we can consider the subsets of $V$ in increasing lexicographic order. 
The comments delimited by ``[['' and ``]]'' justify 
the computations in the algorithm.

\medskip\noindent{\bf Algorithm~U.} (\emph{Up-step.})
Computes the values $F(X,q,i)$ associated with $X$ 
using the values associated with 
$X\setminus\{i\}$ for all $i\in X$.\\
{\em Input:} A subset $X\subseteq V$
     and the value $F(X\setminus\{i\},q,i-1)$ 
             for each $i\in X$ and $q=1,2,\ldots,n+1$.\\
{\em Output:} The value $F(X,q,i)$ for each $q=1,2,\ldots,n+1$
              and $i=0,1,\ldots,n$.
\begin{itemize}
  \item[{\bf U1:}]
  For each $q=1,2,3,\ldots,n+1$, set 
  \[
    F(X,q,n)=\biggl(f(X)+\sum_{i\in X}F(X\setminus\{i\},1,i-1)\biggr)^q\,.
  \]
  \begin{comm}
  By the suffix partition lemma, 
  $\sum_{Y\subsetneq X} f(Y)=\sum_{i\in X}F(X\setminus\{i\},1,i-1)$.
  Adding $f(X)$ and taking powers, we obtain $F(X,q,n)$.
  \end{comm}
  \item[{\bf U2:}] 
  For each $q=1,2,3,\ldots,n+1$ and $i=n,n-1,\ldots,1$, set
  \[
    F(X,q,i-1)=F(X,q,i)-[i\in X]F(X\setminus\{i\},q,i-1)\,.
  \]
  \begin{comm}
    There are two cases to consider to justify correctness.
    First, assume that $i\notin X$. Consider an arbitrary
    $q$-tuple $(U_1,U_2,\ldots,U_q)$ with $U_1,U_2,\ldots,U_q\subseteq X$.
    Let $Y=\cup_{j=1}^q U_j$. Clearly, $Y\subseteq X$. 
    Because $i\notin X$ and $Y\subseteq X$, we have $Y\myieq{i-1} X$ 
    if and only if $Y\myieq{i} X$. Thus, $F(X,q,i-1)=F(X,q,i)$.
    Second, assume that $i\in X$. In this case we have
    $Y\myieq{i} X$ if and only if either
    $Y\myieq{i-1} X$ or $Y\myieq{i-1} X\setminus\{i\}$
    (the former case occurs if $i\in Y$, the latter if $i\notin Y$).
    In the latter case, $Y\subseteq X\setminus\{i\}$
    and hence $U_1,U_2,\ldots,U_q\subseteq X\setminus\{i\}$.
    Thus, $F(X,q,i-1)=F(X,q,i)-F(X\setminus\{i\},q,i-1)$.
  \end{comm}
\end{itemize}

Assume that $f$ satisfies the following property:
for all $X\subseteq V$ it holds that 
\begin{equation}
\label{eq:component-factorisation}
f(X)=f(X_1)f(X_2)\cdots f(X_s)
\end{equation}
where $\myI\mygG {X_1},\myI\mygG {X_2},\ldots,\myI\mygG {X_s}$
are the connected components of $\myI\mygG X$.
For convenience we also assume that $f(\emptyset)=1$.
Note that the factorisation \eqref{eq:component-factorisation}
is well-defined because of commutativity of $R$.
Also note that \eqref{eq:potts-f} satisfies 
\eqref{eq:component-factorisation}.

\begin{Lem}
\label{lem:component}
Let $\myI\mygG {X_1},\myI\mygG {X_2},\ldots,\myI\mygG {X_s}$
be the connected components of $\myI\mygG X$. 
Then,
\begin{equation}
\label{eq:f-component}
F(X,q,i)=\prod_{k=1}^{s} F(X_k,q,i)\,.
\end{equation}
\end{Lem}

The recursion (\ref{eq:f-component}) now enables the following 
top-down evaluation strategy for the intermediate values 
in Definition \ref{def:fpoly}. Consider a nonempty $X\subseteq V$.
If $\myI\mygG X$ is not connected, recursively solve 
the intermediate values of each of the vertex sets 
$X_1,X_2,\ldots,X_s$ of the connected components 
$\myI\mygG {X_1},\myI\mygG {X_2},\ldots,\myI\mygG {X_s}$
of $\myI\mygG X$, and assemble 
the solution using (\ref{eq:f-component}). 
Otherwise; that is, if $\myI\mygG X$ is connected, 
recursively solve the intermediate values of each set
$X\setminus\{i\}$, $i\in X$, and assemble the solution 
using Algorithm~U. Call this evaluation strategy Algorithm~C.

Algorithm~C together with \eqref{eq:multivariate-eval-classical}, 
\eqref{eq:fortuin--kasteleyn}, and Lagrangian interpolation 
establishes that the Tutte polynomial $T_G(x,y)$ can be computed 
in time and space $\sigma(G)n^{O(1)}$. 
This proves Theorem \ref{thm:connected-main}(a).

\subsection{An alternative recursion}
\label{section:recurrence}

We derive an alternative recursion for $Z_G(q,r)$ based 
on induced subgraphs and fast subset convolution. 
Let $R$ be a commutative ring. Associate a ring element
$r_e\in R$ with each $e\in E$.
For $k=1,2,\ldots,n$, let
\[
S_G(k,r)=\sum_{\substack{F\subseteq E\\c(F)=k}}\prod_{e\in F}r_e
\]
and observe that $Z_G(q,r)=\sum_{k=1}^n q^kS_G(k,r)$.
Thus, to determine $Z_G(q,r)$, it suffices to compute $S_G(k,r)$
for all $k=1,2,\ldots,n$.

To this end, the values $S_G(k,r)$ can be computed using 
the following recursion over induced subgraphs of $G$.
Let $W\subseteq V$ and consider the subgraph $G[W]$ induced by $W$ in $G$.
Suppose that $S_{G[U]}(k,r)$ has been computed for all 
$\emptyset\neq U\subsetneq W$ and $k=1,2,\ldots,|U|$. 

To compute $S_{G[W]}(k,r)$ for $k=2,3,\ldots,|W|$, observe that 
a disconnected subgraph of $G[W]$ partitions into connected components.
Thus, for $k\geq 2$ we have
\begin{equation}
\label{eq:alt-rec-1}
S_{G[W]}(k,r)=
\frac{1}{k}\sum_{\emptyset\neq U\subsetneq W} 
S_{G[U]}(1,r)S_{G[W\setminus U]}(k-1,r)\,.
\end{equation}

For the connected case, that is, for $k=1$, it suffices to observe that 
we can subtract the disconnected subgraphs from the set of all subgraphs
to obtain the connected graphs; put otherwise,
\begin{equation}
\label{eq:alt-rec-2}
S_{G[W]}(1,r)=\!\!\prod_{e\in E(G[W])}\!\!(1+r_e)-\sum_{k\geq 2} S_{G[W]}(k,r)\,.
\end{equation}

The recursion defined by \eqref{eq:alt-rec-1} and \eqref{eq:alt-rec-2}
can now be evaluated for $|W|=1,2,\ldots,n$ in total $2^nn^{O(1)}$
ring operations using fast subset convolution \cite{BHKK07a}.  As a
technical observation we remark that \eqref{eq:alt-rec-1} assumes that
$k$ has a multiplicative inverse in $R$; this assumption can be
removed, but we omit the details from this extended abstract. We also
note that analogues of Algorithms U and C running in
$\sigma(G)n^{O(1)}$ ring operations can be developed in this context;
we describe an implementation of this in
Appendix~\ref{appendix:implementation}.  However, it is not
immediate whether a polynomial-space algorithm for the Tutte
polynomial can be developed based on \eqref{eq:alt-rec-1} and
\eqref{eq:alt-rec-2}.


\newpage
\renewcommand{\thepage}{\textsc{References~p.~\arabic{page}}}
\setcounter{page}{1}


\newpage
\renewcommand{\thepage}{\textsc{Appendix~p.~\arabic{page}}}
\appendix
\setcounter{page}{1}

\noindent
\begin{center}
\textsc{\Large Appendix}
\end{center}

\section{Proofs}
\label{appendix:proofs}

\subsection{Proof of Theorem \ref{thm:fortuin--kasteleyn}}

This proof of the Fortuin--Kasteleyn identity 
\eqref{eq:fortuin--kasteleyn} is well known (e.g.~\cite{Soka05})
and is here included only for convenience of verification.

\begin{Proof}
Expanding the product over $E$ and 
changing the order of summation,
\[
Z_G^{\text{Potts}}(q,r)=
\sum_{s:V\rightarrow \{1,2,\ldots,q\}}
\prod_{e\in E}\bigl(1+r_e\delta_e^s\bigr)=
\sum_{F\subseteq E}
\sum_{s:V\rightarrow \{1,2,\ldots,q\}}
\prod_{e\in F}r_e\delta_e^s\,.
\]
The right-hand side product evaluates to zero unless
$s$ is constant on each connected component of the graph
with vertex set $V$ and edge set $F$. Because there are 
$q$ choices for the value of $s$ on each connected component,
\[
\sum_{F\subseteq E}
\sum_{s:V\rightarrow \{1,2,\ldots,q\}}
\prod_{e\in F}r_e\delta_e^s=
\sum_{F\subseteq E} q^{c(F)} \prod_{e\in F}r_e=Z_G(q,r)\,.
\]
\end{Proof}

\subsection{Proof of Lemma \ref{lem:component}}

It is convenient to start with a preliminary lemma.

\begin{Lem}
\label{lem:u-split}
Let $\myI\mygG {X_1},\myI\mygG {X_2},\ldots,\myI\mygG {X_s}$
be the connected components of $\myI\mygG X$
and let $U\subseteq X$.
Then,
\[
f(U)=f(U\cap X_1)f(U\cap X_2)\cdots f(U\cap X_s)\,.
\]
\end{Lem}

\begin{proof}
Let $G[U_1],G[U_2],\ldots,G[U_t]$ be the connected components
of $G[U]$. Then, by \eqref{eq:component-factorisation},
\[
f(U)=f(U_1)f(U_2)\cdots f(U_t)\,.
\]
Because $U\subseteq X$ holds, for every $U_i$ 
there is a unique $h(i)\in\{1,2,\ldots,s\}$ 
such that $U_i\subseteq X_{h(i)}$. 
Moreover, since $\{U_1,U_2,\ldots,U_t\}$ is a partition of $U$,
we have that $\{U_i:i\in h^{-1}(j)\}$ is a partition of $U\cap X_j$ 
for all $j=1,2,\ldots,s$.
Thus, by \eqref{eq:component-factorisation} we have 
$f(U\cap X_j)=\prod_{i\in h^{-1}(j)} f(U_i)$ for all $j=1,2,\ldots,s$.
In particular, by commutativity of $R$,
\[
f(U)
=\prod_{i=1}^t f(U_i)
=\prod_{j=1}^s\prod_{i\in h^{-1}(j)} f(U_i)
=\prod_{j=1}^sf(U\cap X_j)\,.
\]
\end{proof}

We now proceed with the proof of Lemma \ref{lem:component}.

\begin{proof}
Consider an arbitrary $q$-tuple $(U_1,U_2,\ldots,U_q)$ 
with $U_1,U_2,\ldots,U_q\subseteq X$ and
$\cup_{j=1}^q U_j\myieq{i} X$.
Because $\{X_1,X_2,\ldots,X_s\}$ is a partition of $X$,
we have $\cup_{j=1}^q U_j\myieq{i} X$
if and only if 
$X_k\cap\cup_{j=1}^q U_j\myieq{i} X_k\cap X$ holds
for all $k=1,2,\ldots,s$. Put otherwise, 
we have $\cup_{j=1}^q U_j\myieq{i} X$ if and only if
$\cup_{j=1}^q (X_k\cap U_j)\myieq{i} X_k$ holds
for all $k=1,2,\ldots,s$.
Using Lemma \ref{lem:u-split} for each $U_j$ in turn,
we have, by commutativity of $R$, the unique
factorisation into pairwise intersections 
\[
f(U_1)f(U_2)\ldots f(U_q)=
\prod_{j=1}^q\prod_{k=1}^s f(U_j\cap X_k)=
\prod_{k=1}^s\prod_{j=1}^q f(U_j\cap X_k)\,.
\]
The claim follows because $(U_1,U_2,\ldots,U_q)$ was arbitrary.
\end{proof}

\section{A time--space tradeoff via split transforms}
\label{appendix:split-trans}

This appendix outlines a ``split transform'' algorithm that enables
a time--space tradeoff in evaluating $\bigl((f\zeta)^q\mu\bigr)(V)$
for a given function $f:2^V\rightarrow R$ and $q=1,2,\ldots,n+1$.

Split the ground set $V=\{1,2,\ldots,n\}$ into two parts,  
$V_1\subseteq V$ and $V_2\subseteq V$,
such that $V=V_1\cup V_2$ and $V_1\cap V_2=\emptyset$.
Let $n_1=|V_1|$ and $n_2=|V_2|$.
For a subset $X\subseteq V$, we use subscripts to indicate the parts
of the subset in $V_1$ and $V_2$; that is, we let
$X_1=X\cap V_1$ and $X_2=X\cap V_2$.
It is also convenient to split the function notation accordingly,
that is, we write $f(X_1,X_2)$ for $f(X_1\cup X_2)=f(X)$.
In the context of zeta and Moebius transforms, we use $X$ 
for a subset in the ``spatial'' (original) domain and $Y$ 
for a subset in the ``frequency'' (transformed) domain.

An elementary observation is now that both the zeta and Moebius
transforms split, that is, 
\[
f\zeta(Y)=
\sum_{X\subseteq Y}f(X)=
\sum_{X_1\subseteq Y_1}\sum_{X_2\subseteq Y_2}f(X_1,X_2)=
\sum_{X_1\subseteq Y_1}f\zeta_2(X_1,Y_2)=
f\zeta_2\zeta_1(Y_1,Y_2)
\]
and
\[
\begin{split}
f\mu(X)
&=\sum_{Y\subseteq X}(-1)^{|X\setminus Y|}f(Y)
=\sum_{X_1\subseteq Y_1}(-1)^{|X_1\setminus Y_1|}
 \sum_{X_2\subseteq Y_2}(-1)^{|X_2\setminus Y_2|}
 f(Y_1,Y_2)\\
&=\sum_{X_1\subseteq Y_1}(-1)^{|X_1\setminus Y_1|}f\mu_2(Y_1,X_2)
=f\mu_2\mu_1(X_1,X_2)\,.
\end{split}
\]
Also note that $f\zeta=f\zeta_2\zeta_1=f\zeta_1\zeta_2$ 
and $f\mu=f\mu_2\mu_1=f\mu_1\mu_2$.

To arrive at the split transform algorithm for computing
$\bigl((f\zeta)^q\mu\bigr)(V)$, split the outer Moebius transform
and the inner zeta transform to get
\begin{align*}
  \bigl((f\zeta)^q\mu\bigr)(V)
  &=\sum_{Y_1\subseteq V_1}(-1)^{|V_1\setminus Y_1|}\sum_{Y_2\subseteq V_2}(-1)^{|V_2\setminus Y_2|}(f\zeta_1\zeta_2(Y_1,Y_2))^q\enspace.
\end{align*}
Now let $Y_1$ be fixed and consider the inner sum.
To evaluate the inner sum for a fixed $Y_1$, it suffices
to have $f\zeta_1\zeta_2(Y_1,Y_2)$ available for each $Y_2\subseteq V_2$.
By definition,
\[
  f\zeta_1\zeta_2(Y_1,Y_2)=\sum_{X_2\subseteq Y_2}f\zeta_1(Y_1,X_2)\,.
\]
Observe that if we have $f\zeta_1(Y_1,X_2)$ stored for each 
$X_2\subseteq V_2$, then we can evaluate $f\zeta_1\zeta_2(Y_1,Y_2)$ 
for each $Y_2\subseteq V_2$ simultaneously using the fast zeta transform.
This takes in total at most $2^{n_2}n_2$ ring operations and 
requires one to store at most $2^{n_2}n_2$ ring elements.

For fixed $Y_1$ and $X_2$, we can evaluate and store
\[
  f\zeta_1(Y_1,X_2)=\sum_{X_1\subseteq Y_1}f(X_1,X_2)
\]
by plain summation in at most $2^{|Y_1|}$ ring operations.
Thus, for fixed $Y_1$, we can evaluate $f\zeta_1(Y_1,X_2)$ 
for each $X_2\subseteq V_2$ in total at most 
$2^{|Y_1|}2^{n_2}$ ring operations.

Considering each $Y_1\subseteq V_1$ in turn, 
we can thus evaluate $\bigl((f\zeta)^q\mu\bigr)(V)$ by
storing at most $2^{n_2}n_2$ ring
elements and executing at most
\[
  n^{O(1)}\sum_{Y_1\subseteq V_1}(2^{n_2}n_2+2^{|Y_1|}2^{n_2})
   =n^{O(1)}(3^{n_1}+2^{n_1}n_2)2^{n_2}
\]
ring operations. 
This completes the description and analysis of 
the split transform algorithm.

The split transform algorithm together with 
\eqref{eq:multivariate-eval-classical}, 
\eqref{eq:fortuin--kasteleyn}, and Lagrangian interpolation 
proves Theorem \ref{thm:connected-main}(c).

\section{The cover polynomial}
\label{sec: cover}

Let $D$ be a digraph with vertex set $V=\{1,2,\ldots,n\}$.
Note that $D$ may have parallel edges and loops. 
We assume that the number of edges is $n^{O(1)}$.
Denote by $c_D(i,j)$ the number of ways of disjointly
covering all the vertices of $D$ with $i$ directed paths
and $j$ directed cycles.
The \emph{cover polynomial} is defined as
\[
C_D(x,y)=\sum_{i,j}c_D(i,j)x^{\underline i}y^j\,,
\] 
where $x^{\underline i}=x(x-1)\cdots(x-i+1)$ and $x^{\underline 0}=1$.
It is known that $C_D(x,y)$ is
\#P-complete to evaluate except at a handful of points $(x,y)$ \cite{BD07}.

In analogy to Theorem~\ref{thm:connected-main}, we can show that $C_D$
can be computed in vertex-exponential time:

\vspace*{-0.1mm}
\begin{Thm}
\label{thm:cover}
  The cover polynomial of an $n$-vertex directed graph can be computed 
\begin{enumerate}
\item[(a)] in time and space $2^nn^{O(1)}$; and
\item[(b)] in time $3^nn^{O(1)}$ and polynomial space.
\end{enumerate}
\end{Thm}
\vspace*{-0.1mm}

The proof involves several inclusion--exclusion-based arguments with
different purposes and in a nested fashion, so we first give a
high-level overview of the concepts involved. One readily observes
that the cover polynomial can be expressed as a sum over partitionings
of the vertex set, each vertex subset appropriately weighted, so the
inclusion--exclusion technique \cite{BHK07} applies. Computing the
weights for all possible vertex subsets is again a hard problem, but
the fast Moebius inversion algorithm \cite{BHKK08} can be used to
compute the necessary values beforehand.  This leads to an
exponential-space algorithm.  Finally, to use inclusion--exclusion to
reduce the space to polynomial \cite{Karp82,KGK77}, we
apply the mentioned transforms in a nested manner and switch the order
of certain involved summations.

\medskip We turn to the details of the proof.
For $X\subseteq V$, denote by $p(X)$ the number of spanning
directed paths in $D[X]$, and denote by $c(X)$ 
the number of spanning directed cycles in $D[X]$.
Define $p(\emptyset)=c(\emptyset)=0$.
Note that for all $x\in V$ we have $p(\{x\})=1$ 
and that $c(\{x\})$ is the number of loops incident with $x$.

By definition, 
\[
	c_D(i,j)=\frac{1}{i!j!}\sum_{X_1, X_2, \ldots, X_i, Y_1, Y_2, \ldots, Y_j}
	p(X_1)p(X_2) \cdots p(X_i)\, c(Y_1)c(Y_2) \cdots c(Y_j)\,,
\]
where we sum over all $(i+j)$-tuples   
$(X_1, X_2, \ldots, X_i, Y_1, Y_2, \ldots, Y_j)$ 
such that $\{X_1, X_2, \ldots, X_i,\allowbreak Y_1, Y_2, \ldots, Y_j\}$ 
is a partition of $V$.

We next derive an alternative expression  using the principle of inclusion 
and exclusion. To this end, it is convenient to define for every 
$U \subseteq V$ the polynomials  
\[
	P(U; z) = \sum_{X \subseteq U} p(X)z^{|X|}
	\quad
	\textrm{and}
	\quad
	C(U; z) = \sum_{X \subseteq U} c(X)z^{|X|}
\]
in an indeterminate $z$;  
if viewed as set functions, $P(U; z)$ and $C(U; z)$ 
are zeta transforms of the set functions $p(X) z^{|X|}$ and 
$c(X) z^{|X|}$, respectively. 
We can now write 
\[
	c_D(i,j)=\frac{1}{i!j!}\sum_{U \subseteq V} 
	(-1)^{|V\setminus U|} \bigl\{z^n\bigr\}
	\bigl(P(U; z)^i C(U; z)^j\bigr)\,.
\]

It remains to show how to compute the $p(X)$ and $c(X)$ for all $X \subseteq V$.
For $S \subseteq V$ let $w(S, s, t, \ell)$ denote the number of directed walks
of length $\ell$ from vertex $s$ to vertex $t$ in $D[S]$; define  
$w(S, s, t, \ell) = 0$ if $s \not\in S$ or $t \not\in S$. 
By inclusion--exclusion, again,
\[
	p(X) = \sum_{1\leq s \leq t \leq n} \sum_{S \subseteq X} 
	(-1)^{|X \setminus S|} w(S, s, t, |X|-1)\,.
\]
Similarly, 
\[
	c(X) = \sum_{S \subseteq X} 
	(-1)^{|X \setminus S|} w(S, s, s, |X|)\,, 
	\quad \textrm{where }\, s = \min S\,.
\]

Observing that $w(S,s,t,\ell)$ can be computed in time $n^{O(1)}$,
we have that $c_D(i,j)$ can be computed in space $n^{O(1)}$
and time $4^nn^{O(1)}$.

To get an algorithm running in $3^nn^{O(1)}$ time and $n^{O(1)}$ space, 
observe that 
\[
  P(U;z)=\sum_{S\subseteq U} P(U,S;z)
\]
where
\[
  P(U,S;z)=\sum_{1\leq s\leq t\leq n}\sum_{k=0}^{|U\setminus S|}\binom{|U\setminus S|}{k}(-1)^kz^{|S|+k} w(S,s,t,|S|+k-1)
\]
and
\[
  C(U;z)=\sum_{S\subseteq U} C(U,S;z)
\]
where
\[
  C(U,S;z)=\sum_{k=0}^{|U\setminus S|}\binom{|U\setminus S|}{k}(-1)^kz^{|S|+k} w(S,s,s,|S|+k)\,, 
	\quad \textrm{where }\, s = \min S\,.
\]
This establishes part (b) of the theorem.

For part (a), we show how to evaluate $c_D(i,j)$ in time and
space $2^nn^{O(1)}$. Namely, $p$ and $c$ can be computed in time and
space $2^nn^{O(1)}$ via fast Moebius inversion. Given $p$ and $c$, the
polynomials $P$ and $C$ can be computed in time and space
$2^nn^{O(1)}$ via fast zeta transform. And finally, given $P$ and $C$,
the inclusion--exclusion expression of $c_D(i,j)$ can be evaluated in
time $2^nn^{O(1)}$.

\section{Tutte polynomials of concrete graphs}
\label{appendix:implementation}

\setcounter{MaxMatrixCols}{25}

\subsection{Algorithm implementation}

Our implementation of the algorithm described in
\S\ref{section:recurrence} uses a number of extra techniques to reduce
the polynomial factors in the time and memory requirements.  In what
follows we assume that $\mygG$ is a connected graph.
\begin{enumerate}
\item
The coefficients $t_{ij}$ of the Tutte polynomial are 
computed modulo a small integer $p$; the computation is repeated 
for sufficiently many different (pairwise coprime) $p$ to enable
recovery of the coefficients via the Chinese Remainder Theorem. 
The number of different $p$ required is determined based on the
available word length and using $\tau(\mygG)$ (computed via the
Matrix--Tree Theorem) as an upper bound for the coefficients.
\item
To save a factor of $m$ in memory, instead of direct computation 
with bivariate polynomials, we compute with univariate 
evaluations of the polynomials at $z=0,1,\ldots,m$, and finally recover 
only the necessary bivariate polynomials from the evaluations via 
Lagrange interpolation. 
\item To save a further factor of $n^2$ in memory, we execute the
  analogue of Algorithm~U for subsets $X$ in a specific order, namely
  in the lexicographic order.  This enables efficient ``in-place''
  computation of the polynomials $F(X,k,i)$ so that, for each $X$, the
  polynomials $F(X,k,i)$ need to be stored only for one value of $i$
  at the time.  Furthermore, we never need all $F(X,k,i)$ for
  $k=2,3,\ldots,n$ explicitly, only a linear combination of them, so
  we count with this instead; however, we omit the details in this
  abstract.
\end{enumerate}

The source code of the algorithm implementation is available by request.
The implementation uses the GNU Multiple Precision Arithmetic
library $\langle$\verb|http://gmplib.org/|$\rangle$ for computation
with large integers. The computed coefficients $t_{ij}$ are checked 
for consistency by verifying that $\sum_{i,j} t_{ij} = \tau(\mygG)$ 
and that $\sum_{i,j} 2^{i+j} t_{ij} = 2^m$.

\subsection{Performance}

The current algorithm implementation uses roughly $2^{n+1}n$ words of
memory for an $n$-vertex graph, which presents a basic obstacle to
practical performance. For example, the practical limit is at $n=25$,
assuming 32 GB of main memory and 64-bit words. This makes our
polynomial space and time--space tradeoff algorithms from
Theorem~\ref{thm:connected-main}(b,c) interesting also from a
practical perspective. At the time of writing, we have implemented the
former, but not yet performed large-scale experiments with it.

In terms of running time, the complete graph $K_n$ presents the worst
case for $n$-vertex inputs for our algorithm. On a 3.66GHz Intel Xeon
CPU with 1MB cache, computing the Tutte polynomial of $K_{17}$ takes
less than an hour, $K_{18}$ takes about three hours, and $K_{22}$
takes 96 hours.  In comparison, both deletion--contraction and
spanning tree enumeration cease to be practical well below this; for
example, $\tau(K_{22})=705429498686404044207947776$ and
$\tau(K_{16})=72057594037927936$; a survey of how to compute $T_\mygG$
in practice \cite{Imai00} reports running times for the complete graph
$K_{14}$ in hours. The fastest current program to compute Tutte
polynomials \cite{HPR} is also based on deletion--contraction with
isomorphism rejection, but uses many other ideas as well. It processes
$K_{14}$ and many sparse graphs with far larger $n$ in a few seconds,
but also ceases to be practical for some dense graphs with $n=16$, see
Figure~\ref{fig: bhkk-vs-hpr}.

\begin{figure}
  \includegraphics{bhkk-vs-hpr.1}
  \caption{\label{fig: bhkk-vs-hpr}Running times for complements
    of random 4-regular graphs. The lines show averages of 5 runs on a
    3.66GHz Intel Xeon CPU with 1MB cache.  
    The thin line is our algorithm; the thick
    line is the algorithm of Haggard, Pearce, and Royle \cite{HPR}.}
\end{figure}

Two further remarks are in order.  First, for (connected) graphs with
a small $\tau(\mygG)$, enumeration of spanning trees is faster than
our algorithm.  Second, graphs with fewer edges are faster to solve
using our algorithm. For example, a 3-regular graph on 22 vertices can
be solved in about five hours.

\subsection{Tutte polynomials of some concrete graphs}

Even though few readers are likely to derive any insight from the fact
that the coefficient of $x^2y^2$ in the Tutte polynomial of
Loupekine's Second Snark is $991226$, we feel it germane to our paper
to actually compute some Tutte polynomials.  We include tables of the
nonzero coefficients $t_{ij}$ in the
expansion~\eqref{eq:spanning-tree-expansion} for a number of graphs.
Among these, the values for the Petersen graph are well known
\cite[\S13b]{Bigg93} and are included here for verification only.  For
reference, we present the Tutte polynomials of a few other well-known
graphs, mostly snarks and cages; however, these graphs are fairly
sparse and exhibit symmetries that make them amenable to many of the
previously existing techniques.  An entertaining example that tests
the liminations of our current implementation is from Knuth's Stanford
Graph Base \cite{SGB}, based on the encounters between the 23 most
important characters in Twain's \emph{Huckleberry Finn}.  This graph
has 23 vertices, 88 edges, and 54540490752786432 spanning trees; the
required solution time is about 50 hours.

\newcommand{\coefficients}[3]{%
\par\medskip
{\bf #2}\ \par\nopagebreak
\begin{minipage}[t]{.2\textwidth}\vspace{0pt}
\includegraphics{graphs.#1}
\end{minipage}
{\fontsize{5}{6}\selectfont \sf
\input{data/#3.tex}
}
}

\raggedright
\raggedbottom
\par\medskip
{\bf Petersen Graph}\ \par\nopagebreak
\begin{minipage}[t]{.2\textwidth}\vspace{0pt}
\includegraphics{graphs.7}
\end{minipage}
{\fontsize{5}{6}\selectfont \sf
\begin{tabular}[t]{r}\sl j =\,0\\\hline
0
\\
36
\\
120
\\
180
\\
170
\\
114
\\
56
\\
21
\\
6
\\
1
\\
\end{tabular}\allowbreak
\begin{tabular}[t]{r}\sl j =\,1\\\hline
36
\\
168
\\
240
\\
170
\\
70
\\
12
\\
\end{tabular}\allowbreak
\begin{tabular}[t]{r}\sl j =\,2\\\hline
84
\\
171
\\
105
\\
30
\\
\end{tabular}\allowbreak
\begin{tabular}[t]{r}\sl j =\,3\\\hline
75
\\
65
\\
15
\\
\end{tabular}\allowbreak
\begin{tabular}[t]{r}\sl j =\,4\\\hline
35
\\
10
\\
\end{tabular}\allowbreak
\begin{tabular}[t]{r}\sl j =\,5\\\hline
9
\\
\end{tabular}\allowbreak
\begin{tabular}[t]{r}\sl j =\,6\\\hline
1
\\
\end{tabular}\allowbreak

}

\par\medskip
{\bf Dodecahedron}\ \par\nopagebreak
\begin{minipage}[t]{.2\textwidth}\vspace{0pt}
\includegraphics{graphs.1}
\end{minipage}
{\fontsize{5}{6}\selectfont \sf
\begin{tabular}[t]{r}\sl j =\,0\\\hline
0
\\
4\,412
\\
25\,714
\\
72\,110
\\
131\,380
\\
176\,968
\\
189\,934
\\
170\,690
\\
132\,920
\\
91\,740
\\
56\,852
\\
31\,792
\\
16\,016
\\
7\,216
\\
2\,871
\\
989
\\
286
\\
66
\\
11
\\
1
\\
\end{tabular}\allowbreak
\begin{tabular}[t]{r}\sl j =\,1\\\hline
4\,412
\\
38\,864
\\
128\,918
\\
245\,880
\\
320\,990
\\
316\,256
\\
250\,692
\\
167\,140
\\
96\,400
\\
48\,710
\\
21\,530
\\
8\,198
\\
2\,610
\\
660
\\
120
\\
12
\\
\end{tabular}\allowbreak
\begin{tabular}[t]{r}\sl j =\,2\\\hline
17\,562
\\
95\,646
\\
218\,682
\\
295\,915
\\
275\,910
\\
193\,791
\\
108\,884
\\
50\,850
\\
19\,980
\\
6\,510
\\
1\,674
\\
306
\\
30
\\
\end{tabular}\allowbreak
\begin{tabular}[t]{r}\sl j =\,3\\\hline
30\,686
\\
115\,448
\\
185\,071
\\
174\,870
\\
112\,365
\\
53\,350
\\
19\,810
\\
5\,870
\\
1\,350
\\
220
\\
20
\\
\end{tabular}\allowbreak
\begin{tabular}[t]{r}\sl j =\,4\\\hline
31\,540
\\
82\,550
\\
90\,860
\\
57\,735
\\
24\,140
\\
7\,175
\\
1\,620
\\
270
\\
30
\\
\end{tabular}\allowbreak
\begin{tabular}[t]{r}\sl j =\,5\\\hline
21\,548
\\
38\,322
\\
27\,825
\\
11\,230
\\
2\,775
\\
468
\\
60
\\
\end{tabular}\allowbreak
\begin{tabular}[t]{r}\sl j =\,6\\\hline
10\,439
\\
12\,046
\\
5\,390
\\
1\,240
\\
140
\\
12
\\
\end{tabular}\allowbreak
\begin{tabular}[t]{r}\sl j =\,7\\\hline
3\,693
\\
2\,542
\\
610
\\
60
\\
\end{tabular}\allowbreak
\begin{tabular}[t]{r}\sl j =\,8\\\hline
950
\\
330
\\
30
\\
\end{tabular}\allowbreak
\begin{tabular}[t]{r}\sl j =\,9\\\hline
170
\\
20
\\
\end{tabular}\allowbreak
\begin{tabular}[t]{r}\sl j =\,10\\\hline
19
\\
\end{tabular}\allowbreak
\begin{tabular}[t]{r}\sl j =\,11\\\hline
1
\\
\end{tabular}\allowbreak

}

\par\medskip
{\bf Icosahedron}\ \par\nopagebreak
\begin{minipage}[t]{.2\textwidth}\vspace{0pt}
\includegraphics{graphs.11}
\end{minipage}
{\fontsize{5}{6}\selectfont \sf
\begin{tabular}[t]{r}\sl j =\,0\\\hline
0
\\
4\,412
\\
17\,562
\\
30\,686
\\
31\,540
\\
21\,548
\\
10\,439
\\
3\,693
\\
950
\\
170
\\
19
\\
1
\\
\end{tabular}\allowbreak
\begin{tabular}[t]{r}\sl j =\,1\\\hline
4\,412
\\
38\,864
\\
95\,646
\\
115\,448
\\
82\,550
\\
38\,322
\\
12\,046
\\
2\,542
\\
330
\\
20
\\
\end{tabular}\allowbreak
\begin{tabular}[t]{r}\sl j =\,2\\\hline
25\,714
\\
128\,918
\\
218\,682
\\
185\,071
\\
90\,860
\\
27\,825
\\
5\,390
\\
610
\\
30
\\
\end{tabular}\allowbreak
\begin{tabular}[t]{r}\sl j =\,3\\\hline
72\,110
\\
245\,880
\\
295\,915
\\
174\,870
\\
57\,735
\\
11\,230
\\
1\,240
\\
60
\\
\end{tabular}\allowbreak
\begin{tabular}[t]{r}\sl j =\,4\\\hline
131\,380
\\
320\,990
\\
275\,910
\\
112\,365
\\
24\,140
\\
2\,775
\\
140
\\
\end{tabular}\allowbreak
\begin{tabular}[t]{r}\sl j =\,5\\\hline
176\,968
\\
316\,256
\\
193\,791
\\
53\,350
\\
7\,175
\\
468
\\
12
\\
\end{tabular}\allowbreak
\begin{tabular}[t]{r}\sl j =\,6\\\hline
189\,934
\\
250\,692
\\
108\,884
\\
19\,810
\\
1\,620
\\
60
\\
\end{tabular}\allowbreak
\begin{tabular}[t]{r}\sl j =\,7\\\hline
170\,690
\\
167\,140
\\
50\,850
\\
5\,870
\\
270
\\
\end{tabular}\allowbreak
\begin{tabular}[t]{r}\sl j =\,8\\\hline
132\,920
\\
96\,400
\\
19\,980
\\
1\,350
\\
30
\\
\end{tabular}\allowbreak
\begin{tabular}[t]{r}\sl j =\,9\\\hline
91\,740
\\
48\,710
\\
6\,510
\\
220
\\
\end{tabular}\allowbreak
\begin{tabular}[t]{r}\sl j =\,10\\\hline
56\,852
\\
21\,530
\\
1\,674
\\
20
\\
\end{tabular}\allowbreak
\begin{tabular}[t]{r}\sl j =\,11\\\hline
31\,792
\\
8\,198
\\
306
\\
\end{tabular}\allowbreak
\begin{tabular}[t]{r}\sl j =\,12\\\hline
16\,016
\\
2\,610
\\
30
\\
\end{tabular}\allowbreak
\begin{tabular}[t]{r}\sl j =\,13\\\hline
7\,216
\\
660
\\
\end{tabular}\allowbreak
\begin{tabular}[t]{r}\sl j =\,14\\\hline
2\,871
\\
120
\\
\end{tabular}\allowbreak
\begin{tabular}[t]{r}\sl j =\,15\\\hline
989
\\
12
\\
\end{tabular}\allowbreak
\begin{tabular}[t]{r}\sl j =\,16\\\hline
286
\\
\end{tabular}\allowbreak
\begin{tabular}[t]{r}\sl j =\,17\\\hline
66
\\
\end{tabular}\allowbreak
\begin{tabular}[t]{r}\sl j =\,18\\\hline
11
\\
\end{tabular}\allowbreak
\begin{tabular}[t]{r}\sl j =\,19\\\hline
1
\\
\end{tabular}\allowbreak

}

\par\medskip
{\bf Chv\'atal Graph}\ \par\nopagebreak
\begin{minipage}[t]{.2\textwidth}\vspace{0pt}
\includegraphics{graphs.9}
\end{minipage}
{\fontsize{5}{6}\selectfont \sf
\begin{tabular}[t]{r}\sl j =\,0\\\hline
0
\\
1\,994
\\
7\,427
\\
12\,339
\\
12\,360
\\
8\,445
\\
4\,191
\\
1\,559
\\
438
\\
91
\\
13
\\
1
\\
\end{tabular}\allowbreak
\begin{tabular}[t]{r}\sl j =\,1\\\hline
1\,994
\\
12\,782
\\
25\,604
\\
26\,004
\\
15\,865
\\
6\,216
\\
1\,572
\\
240
\\
17
\\
\end{tabular}\allowbreak
\begin{tabular}[t]{r}\sl j =\,2\\\hline
7\,349
\\
25\,969
\\
32\,754
\\
20\,914
\\
7\,568
\\
1\,552
\\
158
\\
4
\\
\end{tabular}\allowbreak
\begin{tabular}[t]{r}\sl j =\,3\\\hline
12\,626
\\
28\,952
\\
24\,116
\\
9\,804
\\
2\,040
\\
184
\\
2
\\
\end{tabular}\allowbreak
\begin{tabular}[t]{r}\sl j =\,4\\\hline
14\,115
\\
22\,250
\\
12\,508
\\
3\,166
\\
319
\\
4
\\
\end{tabular}\allowbreak
\begin{tabular}[t]{r}\sl j =\,5\\\hline
11\,903
\\
13\,164
\\
4\,882
\\
672
\\
17
\\
\end{tabular}\allowbreak
\begin{tabular}[t]{r}\sl j =\,6\\\hline
8\,140
\\
6\,202
\\
1\,386
\\
72
\\
\end{tabular}\allowbreak
\begin{tabular}[t]{r}\sl j =\,7\\\hline
4\,642
\\
2\,292
\\
258
\\
\end{tabular}\allowbreak
\begin{tabular}[t]{r}\sl j =\,8\\\hline
2\,211
\\
636
\\
24
\\
\end{tabular}\allowbreak
\begin{tabular}[t]{r}\sl j =\,9\\\hline
869
\\
120
\\
\end{tabular}\allowbreak
\begin{tabular}[t]{r}\sl j =\,10\\\hline
274
\\
12
\\
\end{tabular}\allowbreak
\begin{tabular}[t]{r}\sl j =\,11\\\hline
66
\\
\end{tabular}\allowbreak
\begin{tabular}[t]{r}\sl j =\,12\\\hline
11
\\
\end{tabular}\allowbreak
\begin{tabular}[t]{r}\sl j =\,13\\\hline
1
\\
\end{tabular}\allowbreak

}

\par\medskip
{\bf Clebsch Graph}\ \par\nopagebreak
\begin{minipage}[t]{.2\textwidth}\vspace{0pt}
\includegraphics{graphs.8}
\end{minipage}
{\fontsize{5}{6}\selectfont \sf
\begin{tabular}[t]{r}\sl j =\,0\\\hline
0
\\
1\,872\,172
\\
7\,870\,034
\\
15\,033\,470
\\
17\,576\,840
\\
14\,236\,468
\\
8\,544\,936
\\
3\,958\,696
\\
1\,451\,495
\\
427\,155
\\
101\,355
\\
19\,283
\\
2\,885
\\
325
\\
25
\\
1
\\
\end{tabular}\allowbreak
\begin{tabular}[t]{r}\sl j =\,1\\\hline
1\,872\,172
\\
15\,110\,476
\\
38\,438\,772
\\
51\,332\,560
\\
43\,215\,300
\\
25\,097\,376
\\
10\,555\,976
\\
3\,293\,168
\\
765\,300
\\
130\,280
\\
15\,488
\\
1\,152
\\
40
\\
\end{tabular}\allowbreak
\begin{tabular}[t]{r}\sl j =\,2\\\hline
9\,112\,614
\\
43\,880\,542
\\
79\,492\,384
\\
78\,503\,860
\\
49\,009\,780
\\
20\,801\,316
\\
6\,205\,768
\\
1\,305\,736
\\
187\,860
\\
16\,860
\\
720
\\
\end{tabular}\allowbreak
\begin{tabular}[t]{r}\sl j =\,3\\\hline
21\,717\,820
\\
75\,108\,240
\\
103\,270\,060
\\
78\,511\,920
\\
37\,661\,120
\\
12\,109\,800
\\
2\,654\,560
\\
386\,480
\\
34\,000
\\
1\,360
\\
\end{tabular}\allowbreak
\begin{tabular}[t]{r}\sl j =\,4\\\hline
34\,847\,530
\\
93\,048\,150
\\
101\,400\,130
\\
61\,562\,510
\\
23\,461\,820
\\
5\,860\,600
\\
946\,240
\\
90\,320
\\
3\,860
\\
\end{tabular}\allowbreak
\begin{tabular}[t]{r}\sl j =\,5\\\hline
43\,384\,468
\\
93\,485\,328
\\
83\,435\,332
\\
41\,488\,560
\\
12\,724\,460
\\
2\,443\,840
\\
275\,672
\\
14\,800
\\
140
\\
\end{tabular}\allowbreak
\begin{tabular}[t]{r}\sl j =\,6\\\hline
45\,431\,208
\\
81\,408\,316
\\
60\,699\,616
\\
24\,896\,600
\\
6\,070\,620
\\
859\,832
\\
61\,328
\\
1\,360
\\
\end{tabular}\allowbreak
\begin{tabular}[t]{r}\sl j =\,7\\\hline
42\,011\,212
\\
63\,725\,936
\\
39\,921\,392
\\
13\,374\,520
\\
2\,514\,620
\\
245\,520
\\
9\,200
\\
\end{tabular}\allowbreak
\begin{tabular}[t]{r}\sl j =\,8\\\hline
35\,302\,105
\\
45\,628\,390
\\
23\,866\,000
\\
6\,392\,880
\\
886\,920
\\
54\,000
\\
720
\\
\end{tabular}\allowbreak
\begin{tabular}[t]{r}\sl j =\,9\\\hline
27\,382\,885
\\
30\,079\,420
\\
12\,946\,480
\\
2\,691\,440
\\
259\,240
\\
8\,240
\\
\end{tabular}\allowbreak
\begin{tabular}[t]{r}\sl j =\,10\\\hline
19\,759\,258
\\
18\,276\,422
\\
6\,341\,100
\\
983\,800
\\
60\,100
\\
672
\\
\end{tabular}\allowbreak
\begin{tabular}[t]{r}\sl j =\,11\\\hline
13\,306\,232
\\
10\,217\,568
\\
2\,783\,100
\\
305\,640
\\
10\,220
\\
\end{tabular}\allowbreak
\begin{tabular}[t]{r}\sl j =\,12\\\hline
8\,367\,140
\\
5\,236\,520
\\
1\,082\,640
\\
78\,080
\\
1\,080
\\
\end{tabular}\allowbreak
\begin{tabular}[t]{r}\sl j =\,13\\\hline
4\,907\,540
\\
2\,446\,480
\\
367\,440
\\
15\,520
\\
40
\\
\end{tabular}\allowbreak
\begin{tabular}[t]{r}\sl j =\,14\\\hline
2\,678\,480
\\
1\,033\,720
\\
106\,320
\\
2\,160
\\
\end{tabular}\allowbreak
\begin{tabular}[t]{r}\sl j =\,15\\\hline
1\,355\,496
\\
390\,712
\\
25\,320
\\
160
\\
\end{tabular}\allowbreak
\begin{tabular}[t]{r}\sl j =\,16\\\hline
632\,942
\\
130\,088
\\
4\,680
\\
\end{tabular}\allowbreak
\begin{tabular}[t]{r}\sl j =\,17\\\hline
270\,930
\\
37\,320
\\
600
\\
\end{tabular}\allowbreak
\begin{tabular}[t]{r}\sl j =\,18\\\hline
105\,400
\\
8\,920
\\
40
\\
\end{tabular}\allowbreak
\begin{tabular}[t]{r}\sl j =\,19\\\hline
36\,840
\\
1\,680
\\
\end{tabular}\allowbreak
\begin{tabular}[t]{r}\sl j =\,20\\\hline
11\,388
\\
224
\\
\end{tabular}\allowbreak
\begin{tabular}[t]{r}\sl j =\,21\\\hline
3\,044
\\
16
\\
\end{tabular}\allowbreak
\begin{tabular}[t]{r}\sl j =\,22\\\hline
680
\\
\end{tabular}\allowbreak
\begin{tabular}[t]{r}\sl j =\,23\\\hline
120
\\
\end{tabular}\allowbreak
\begin{tabular}[t]{r}\sl j =\,24\\\hline
15
\\
\end{tabular}\allowbreak
\begin{tabular}[t]{r}\sl j =\,25\\\hline
1
\\
\end{tabular}\allowbreak

}

\par\medskip
{\bf Brinkmann Graph}\ \par\nopagebreak
\begin{minipage}[t]{.2\textwidth}\vspace{0pt}
\includegraphics{graphs.10}
\end{minipage}
{\fontsize{5}{6}\selectfont \sf
\begin{tabular}[t]{r}\sl j =\,0\\\hline
0
\\
9\,135\,298
\\
49\,413\,533
\\
127\,008\,274
\\
208\,645\,102
\\
247\,964\,242
\\
228\,346\,378
\\
170\,148\,325
\\
105\,629\,121
\\
55\,758\,397
\\
25\,384\,606
\\
10\,061\,144
\\
3\,489\,936
\\
1\,060\,656
\\
281\,455
\\
64\,604
\\
12\,601
\\
2\,024
\\
253
\\
22
\\
1
\\
\end{tabular}\allowbreak
\begin{tabular}[t]{r}\sl j =\,1\\\hline
9\,135\,298
\\
81\,926\,895
\\
266\,495\,740
\\
489\,646\,682
\\
603\,283\,289
\\
545\,064\,597
\\
380\,867\,123
\\
212\,835\,902
\\
97\,136\,821
\\
36\,640\,429
\\
11\,470\,466
\\
2\,969\,085
\\
626\,955
\\
105\,035
\\
13\,239
\\
1\,127
\\
49
\\
\end{tabular}\allowbreak
\begin{tabular}[t]{r}\sl j =\,2\\\hline
41\,648\,660
\\
239\,055\,379
\\
573\,449\,072
\\
809\,702\,257
\\
776\,194\,328
\\
545\,163\,829
\\
293\,316\,401
\\
124\,062\,007
\\
41\,782\,196
\\
11\,223\,611
\\
2\,380\,728
\\
388\,661
\\
46\,361
\\
3\,626
\\
140
\\
\end{tabular}\allowbreak
\begin{tabular}[t]{r}\sl j =\,3\\\hline
91\,803\,040
\\
397\,141\,486
\\
748\,656\,482
\\
842\,806\,153
\\
645\,651\,909
\\
360\,374\,231
\\
152\,129\,831
\\
49\,440\,398
\\
12\,395\,327
\\
2\,363\,508
\\
330\,351
\\
31\,304
\\
1\,680
\\
28
\\
\end{tabular}\allowbreak
\begin{tabular}[t]{r}\sl j =\,4\\\hline
134\,604\,309
\\
465\,831\,800
\\
714\,722\,253
\\
658\,050\,897
\\
410\,955\,629
\\
184\,948\,316
\\
61\,696\,341
\\
15\,337\,490
\\
2\,795\,009
\\
356\,678
\\
28\,840
\\
1\,155
\\
7
\\
\end{tabular}\allowbreak
\begin{tabular}[t]{r}\sl j =\,5\\\hline
151\,187\,372
\\
432\,446\,574
\\
551\,889\,933
\\
421\,620\,731
\\
216\,136\,928
\\
78\,214\,626
\\
20\,281\,751
\\
3\,715\,740
\\
457\,702
\\
33\,824
\\
1\,113
\\
\end{tabular}\allowbreak
\begin{tabular}[t]{r}\sl j =\,6\\\hline
140\,741\,055
\\
338\,941\,057
\\
363\,367\,900
\\
230\,681\,654
\\
96\,269\,502
\\
27\,429\,213
\\
5\,314\,498
\\
668\,668
\\
48\,664
\\
1\,505
\\
\end{tabular}\allowbreak
\begin{tabular}[t]{r}\sl j =\,7\\\hline
113\,681\,473
\\
232\,203\,883
\\
208\,849\,858
\\
109\,043\,291
\\
36\,267\,791
\\
7\,843\,003
\\
1\,066\,233
\\
81\,926
\\
2\,667
\\
7
\\
\end{tabular}\allowbreak
\begin{tabular}[t]{r}\sl j =\,8\\\hline
81\,746\,167
\\
141\,329\,485
\\
105\,525\,301
\\
44\,426\,158
\\
11\,403\,399
\\
1\,773\,968
\\
153\,321
\\
5\,749
\\
28
\\
\end{tabular}\allowbreak
\begin{tabular}[t]{r}\sl j =\,9\\\hline
53\,017\,571
\\
76\,846\,499
\\
46\,789\,365
\\
15\,447\,215
\\
2\,926\,959
\\
302\,344
\\
13\,951
\\
140
\\
\end{tabular}\allowbreak
\begin{tabular}[t]{r}\sl j =\,10\\\hline
31\,181\,857
\\
37\,308\,376
\\
18\,078\,277
\\
4\,510\,996
\\
592\,536
\\
35\,602
\\
588
\\
\end{tabular}\allowbreak
\begin{tabular}[t]{r}\sl j =\,11\\\hline
16\,645\,377
\\
16\,100\,987
\\
6\,016\,780
\\
1\,080\,072
\\
89\,355
\\
2\,401
\\
\end{tabular}\allowbreak
\begin{tabular}[t]{r}\sl j =\,12\\\hline
8\,048\,376
\\
6\,129\,326
\\
1\,695\,687
\\
204\,344
\\
9\,009
\\
49
\\
\end{tabular}\allowbreak
\begin{tabular}[t]{r}\sl j =\,13\\\hline
3\,509\,821
\\
2\,034\,949
\\
394\,632
\\
28\,756
\\
462
\\
\end{tabular}\allowbreak
\begin{tabular}[t]{r}\sl j =\,14\\\hline
1\,371\,591
\\
579\,747
\\
72\,996
\\
2\,688
\\
\end{tabular}\allowbreak
\begin{tabular}[t]{r}\sl j =\,15\\\hline
476\,045
\\
138\,453
\\
10\,080
\\
126
\\
\end{tabular}\allowbreak
\begin{tabular}[t]{r}\sl j =\,16\\\hline
144\,970
\\
26\,754
\\
924
\\
\end{tabular}\allowbreak
\begin{tabular}[t]{r}\sl j =\,17\\\hline
38\,094
\\
3\,948
\\
42
\\
\end{tabular}\allowbreak
\begin{tabular}[t]{r}\sl j =\,18\\\hline
8\,435
\\
399
\\
\end{tabular}\allowbreak
\begin{tabular}[t]{r}\sl j =\,19\\\hline
1\,519
\\
21
\\
\end{tabular}\allowbreak
\begin{tabular}[t]{r}\sl j =\,20\\\hline
210
\\
\end{tabular}\allowbreak
\begin{tabular}[t]{r}\sl j =\,21\\\hline
20
\\
\end{tabular}\allowbreak
\begin{tabular}[t]{r}\sl j =\,22\\\hline
1
\\
\end{tabular}\allowbreak

}

\par\medskip
{\bf McGee Graph}\ \par\nopagebreak
\begin{minipage}[t]{.2\textwidth}\vspace{0pt}
\includegraphics{graphs.6}
\end{minipage}
{\fontsize{5}{6}\selectfont \sf
\begin{tabular}[t]{r}\sl j =\,0\\\hline
0
\\
100\,424
\\
616\,320
\\
1\,853\,724
\\
3\,683\,515
\\
5\,484\,441
\\
6\,563\,798
\\
6\,600\,622
\\
5\,745\,907
\\
4\,420\,661
\\
3\,050\,680
\\
1\,908\,584
\\
1\,090\,666
\\
572\,080
\\
276\,100
\\
122\,600
\\
49\,938
\\
18\,532
\\
6\,188
\\
1\,820
\\
455
\\
91
\\
13
\\
1
\\
\end{tabular}\allowbreak
\begin{tabular}[t]{r}\sl j =\,1\\\hline
100\,424
\\
863\,904
\\
2\,984\,380
\\
6\,149\,836
\\
8\,896\,534
\\
9\,867\,514
\\
8\,854\,364
\\
6\,650\,972
\\
4\,272\,590
\\
2\,377\,190
\\
1\,152\,488
\\
486\,960
\\
178\,182
\\
55\,608
\\
14\,380
\\
2\,920
\\
418
\\
32
\\
\end{tabular}\allowbreak
\begin{tabular}[t]{r}\sl j =\,2\\\hline
348\,008
\\
1\,945\,060
\\
4\,833\,738
\\
7\,421\,464
\\
8\,101\,789
\\
6\,792\,829
\\
4\,577\,890
\\
2\,543\,854
\\
1\,178\,731
\\
455\,693
\\
145\,600
\\
37\,584
\\
7\,488
\\
1\,048
\\
80
\\
\end{tabular}\allowbreak
\begin{tabular}[t]{r}\sl j =\,3\\\hline
546\,092
\\
2\,252\,476
\\
4\,237\,698
\\
4\,969\,160
\\
4\,145\,382
\\
2\,639\,544
\\
1\,330\,356
\\
536\,856
\\
172\,060
\\
42\,648
\\
7\,728
\\
904
\\
48
\\
\end{tabular}\allowbreak
\begin{tabular}[t]{r}\sl j =\,4\\\hline
537\,899
\\
1\,697\,518
\\
2\,455\,880
\\
2\,208\,876
\\
1\,401\,096
\\
666\,144
\\
241\,640
\\
65\,472
\\
12\,502
\\
1\,488
\\
80
\\
\end{tabular}\allowbreak
\begin{tabular}[t]{r}\sl j =\,5\\\hline
382\,951
\\
930\,400
\\
1\,027\,312
\\
694\,880
\\
323\,908
\\
108\,612
\\
25\,476
\\
3\,720
\\
258
\\
\end{tabular}\allowbreak
\begin{tabular}[t]{r}\sl j =\,6\\\hline
210\,826
\\
387\,550
\\
316\,166
\\
153\,316
\\
48\,978
\\
10\,308
\\
1\,188
\\
32
\\
\end{tabular}\allowbreak
\begin{tabular}[t]{r}\sl j =\,7\\\hline
92\,060
\\
122\,924
\\
69\,834
\\
22\,388
\\
4\,326
\\
432
\\
\end{tabular}\allowbreak
\begin{tabular}[t]{r}\sl j =\,8\\\hline
31\,878
\\
28\,908
\\
10\,404
\\
1\,920
\\
168
\\
\end{tabular}\allowbreak
\begin{tabular}[t]{r}\sl j =\,9\\\hline
8\,602
\\
4\,764
\\
924
\\
72
\\
\end{tabular}\allowbreak
\begin{tabular}[t]{r}\sl j =\,10\\\hline
1\,748
\\
492
\\
36
\\
\end{tabular}\allowbreak
\begin{tabular}[t]{r}\sl j =\,11\\\hline
252
\\
24
\\
\end{tabular}\allowbreak
\begin{tabular}[t]{r}\sl j =\,12\\\hline
23
\\
\end{tabular}\allowbreak
\begin{tabular}[t]{r}\sl j =\,13\\\hline
1
\\
\end{tabular}\allowbreak

}

\par\medskip
{\bf Flower Snark}\ \par\nopagebreak
\begin{minipage}[t]{.2\textwidth}\vspace{0pt}
\includegraphics{graphs.4}
\end{minipage}
{\fontsize{5}{6}\selectfont \sf
\begin{tabular}[t]{r}\sl j =\,0\\\hline
0
\\
7\,878
\\
43\,135
\\
114\,690
\\
200\,340
\\
261\,282
\\
273\,073
\\
239\,007
\\
180\,402
\\
119\,792
\\
70\,904
\\
37\,697
\\
18\,052
\\
7\,767
\\
2\,977
\\
1\,000
\\
286
\\
66
\\
11
\\
1
\\
\end{tabular}\allowbreak
\begin{tabular}[t]{r}\sl j =\,1\\\hline
7\,878
\\
61\,874
\\
187\,515
\\
332\,265
\\
407\,935
\\
379\,816
\\
282\,743
\\
173\,800
\\
89\,925
\\
39\,505
\\
14\,713
\\
4\,580
\\
1\,150
\\
215
\\
25
\\
1
\\
\end{tabular}\allowbreak
\begin{tabular}[t]{r}\sl j =\,2\\\hline
26\,617
\\
129\,158
\\
268\,795
\\
336\,780
\\
293\,485
\\
191\,744
\\
97\,651
\\
39\,500
\\
12\,730
\\
3\,220
\\
615
\\
80
\\
5
\\
\end{tabular}\allowbreak
\begin{tabular}[t]{r}\sl j =\,3\\\hline
39\,815
\\
134\,515
\\
198\,500
\\
176\,070
\\
107\,135
\\
47\,405
\\
15\,540
\\
3\,760
\\
640
\\
65
\\
\end{tabular}\allowbreak
\begin{tabular}[t]{r}\sl j =\,4\\\hline
36\,190
\\
86\,880
\\
90\,385
\\
55\,570
\\
22\,700
\\
6\,285
\\
1\,145
\\
125
\\
5
\\
\end{tabular}\allowbreak
\begin{tabular}[t]{r}\sl j =\,5\\\hline
22\,832
\\
38\,396
\\
27\,215
\\
10\,965
\\
2\,715
\\
369
\\
20
\\
\end{tabular}\allowbreak
\begin{tabular}[t]{r}\sl j =\,6\\\hline
10\,624
\\
11\,938
\\
5\,335
\\
1\,240
\\
140
\\
1
\\
\end{tabular}\allowbreak
\begin{tabular}[t]{r}\sl j =\,7\\\hline
3\,704
\\
2\,531
\\
610
\\
60
\\
\end{tabular}\allowbreak
\begin{tabular}[t]{r}\sl j =\,8\\\hline
950
\\
330
\\
30
\\
\end{tabular}\allowbreak
\begin{tabular}[t]{r}\sl j =\,9\\\hline
170
\\
20
\\
\end{tabular}\allowbreak
\begin{tabular}[t]{r}\sl j =\,10\\\hline
19
\\
\end{tabular}\allowbreak
\begin{tabular}[t]{r}\sl j =\,11\\\hline
1
\\
\end{tabular}\allowbreak

}

\par\medskip
{\bf Loupekine's First Snark}\ \par\nopagebreak
\begin{minipage}[t]{.2\textwidth}\vspace{0pt}
\includegraphics{graphs.2}
\end{minipage}
{\fontsize{5}{6}\selectfont \sf
\begin{tabular}[t]{r}\sl j =\,0\\\hline
0
\\
20\,724
\\
121\,838
\\
350\,010
\\
663\,435
\\
941\,666
\\
1\,073\,720
\\
1\,028\,153
\\
852\,031
\\
623\,999
\\
409\,765
\\
243\,580
\\
131\,786
\\
65\,014
\\
29\,187
\\
11\,845
\\
4\,291
\\
1\,359
\\
364
\\
78
\\
12
\\
1
\\
\end{tabular}\allowbreak
\begin{tabular}[t]{r}\sl j =\,1\\\hline
20\,724
\\
176\,872
\\
591\,662
\\
1\,169\,865
\\
1\,617\,133
\\
1\,709\,810
\\
1\,460\,522
\\
1\,044\,783
\\
641\,221
\\
342\,967
\\
161\,263
\\
66\,826
\\
24\,328
\\
7\,703
\\
2\,073
\\
451
\\
71
\\
6
\\
\end{tabular}\allowbreak
\begin{tabular}[t]{r}\sl j =\,2\\\hline
75\,758
\\
412\,502
\\
981\,968
\\
1\,430\,045
\\
1\,470\,305
\\
1\,153\,310
\\
723\,379
\\
373\,894
\\
162\,202
\\
59\,554
\\
18\,481
\\
4\,810
\\
1\,050
\\
195
\\
30
\\
3
\\
\end{tabular}\allowbreak
\begin{tabular}[t]{r}\sl j =\,3\\\hline
124\,770
\\
490\,516
\\
862\,535
\\
930\,845
\\
703\,524
\\
398\,793
\\
176\,425
\\
62\,282
\\
17\,622
\\
3\,891
\\
603
\\
48
\\
\end{tabular}\allowbreak
\begin{tabular}[t]{r}\sl j =\,4\\\hline
126\,268
\\
367\,014
\\
475\,649
\\
373\,695
\\
201\,006
\\
78\,432
\\
22\,916
\\
5\,068
\\
825
\\
81
\\
\end{tabular}\allowbreak
\begin{tabular}[t]{r}\sl j =\,5\\\hline
89\,312
\\
190\,328
\\
177\,299
\\
97\,491
\\
35\,166
\\
8\,688
\\
1\,512
\\
189
\\
15
\\
\end{tabular}\allowbreak
\begin{tabular}[t]{r}\sl j =\,6\\\hline
46\,998
\\
71\,133
\\
45\,315
\\
16\,253
\\
3\,509
\\
465
\\
38
\\
3
\\
\end{tabular}\allowbreak
\begin{tabular}[t]{r}\sl j =\,7\\\hline
18\,864
\\
19\,116
\\
7\,638
\\
1\,568
\\
154
\\
6
\\
\end{tabular}\allowbreak
\begin{tabular}[t]{r}\sl j =\,8\\\hline
5\,766
\\
3\,540
\\
759
\\
66
\\
\end{tabular}\allowbreak
\begin{tabular}[t]{r}\sl j =\,9\\\hline
1\,309
\\
407
\\
33
\\
\end{tabular}\allowbreak
\begin{tabular}[t]{r}\sl j =\,10\\\hline
209
\\
22
\\
\end{tabular}\allowbreak
\begin{tabular}[t]{r}\sl j =\,11\\\hline
21
\\
\end{tabular}\allowbreak
\begin{tabular}[t]{r}\sl j =\,12\\\hline
1
\\
\end{tabular}\allowbreak

}

\par\medskip
{\bf Loupekine's Second Snark}\ \par\nopagebreak
\begin{minipage}[t]{.2\textwidth}\vspace{0pt}
\includegraphics{graphs.3}
\end{minipage}
{\fontsize{5}{6}\selectfont \sf
\begin{tabular}[t]{r}\sl j =\,0\\\hline
0
\\
21\,156
\\
124\,286
\\
356\,730
\\
675\,496
\\
957\,769
\\
1\,090\,933
\\
1\,043\,540
\\
863\,802
\\
631\,780
\\
414\,216
\\
245\,775
\\
132\,710
\\
65\,338
\\
29\,277
\\
11\,863
\\
4\,293
\\
1\,359
\\
364
\\
78
\\
12
\\
1
\\
\end{tabular}\allowbreak
\begin{tabular}[t]{r}\sl j =\,1\\\hline
21\,156
\\
180\,076
\\
601\,016
\\
1\,185\,628
\\
1\,635\,022
\\
1\,724\,581
\\
1\,469\,555
\\
1\,048\,408
\\
641\,304
\\
341\,544
\\
159\,732
\\
65\,772
\\
23\,776
\\
7\,475
\\
2\,001
\\
435
\\
69
\\
6
\\
\end{tabular}\allowbreak
\begin{tabular}[t]{r}\sl j =\,2\\\hline
76\,946
\\
417\,674
\\
991\,226
\\
1\,439\,086
\\
1\,475\,257
\\
1\,154\,068
\\
721\,869
\\
371\,796
\\
160\,464
\\
58\,494
\\
17\,991
\\
4\,641
\\
1\,008
\\
189
\\
30
\\
3
\\
\end{tabular}\allowbreak
\begin{tabular}[t]{r}\sl j =\,3\\\hline
126\,048
\\
493\,856
\\
865\,641
\\
931\,623
\\
702\,666
\\
397\,725
\\
175\,617
\\
61\,744
\\
17\,346
\\
3\,793
\\
585
\\
48
\\
\end{tabular}\allowbreak
\begin{tabular}[t]{r}\sl j =\,4\\\hline
126\,968
\\
367\,930
\\
475\,731
\\
373\,263
\\
200\,712
\\
78\,312
\\
22\,816
\\
5\,000
\\
801
\\
75
\\
\end{tabular}\allowbreak
\begin{tabular}[t]{r}\sl j =\,5\\\hline
89\,520
\\
190\,386
\\
177\,183
\\
97\,431
\\
35\,160
\\
8\,682
\\
1\,494
\\
183
\\
15
\\
\end{tabular}\allowbreak
\begin{tabular}[t]{r}\sl j =\,6\\\hline
47\,030
\\
71\,117
\\
45\,303
\\
16\,253
\\
3\,509
\\
465
\\
36
\\
3
\\
\end{tabular}\allowbreak
\begin{tabular}[t]{r}\sl j =\,7\\\hline
18\,866
\\
19\,114
\\
7\,638
\\
1\,568
\\
154
\\
6
\\
\end{tabular}\allowbreak
\begin{tabular}[t]{r}\sl j =\,8\\\hline
5\,766
\\
3\,540
\\
759
\\
66
\\
\end{tabular}\allowbreak
\begin{tabular}[t]{r}\sl j =\,9\\\hline
1\,309
\\
407
\\
33
\\
\end{tabular}\allowbreak
\begin{tabular}[t]{r}\sl j =\,10\\\hline
209
\\
22
\\
\end{tabular}\allowbreak
\begin{tabular}[t]{r}\sl j =\,11\\\hline
21
\\
\end{tabular}\allowbreak
\begin{tabular}[t]{r}\sl j =\,12\\\hline
1
\\
\end{tabular}\allowbreak

}

\par\medskip
{\bf Robertson Graph}\ \par\nopagebreak
\begin{minipage}[t]{.2\textwidth}\vspace{0pt}
\includegraphics{graphs.5}
\end{minipage}
{\fontsize{5}{6}\selectfont \sf
\begin{tabular}[t]{r}\sl j =\,0\\\hline
0
\\
1\,437\,372
\\
7\,246\,700
\\
17\,211\,692
\\
25\,936\,913
\\
28\,091\,119
\\
23\,425\,656
\\
15\,702\,294
\\
8\,704\,413
\\
4\,067\,425
\\
1\,622\,042
\\
555\,756
\\
163\,804
\\
41\,322
\\
8\,801
\\
1\,540
\\
210
\\
20
\\
1
\\
\end{tabular}\allowbreak
\begin{tabular}[t]{r}\sl j =\,1\\\hline
1\,437\,372
\\
12\,029\,428
\\
35\,805\,218
\\
59\,406\,320
\\
65\,327\,035
\\
52\,063\,835
\\
31\,670\,274
\\
15\,160\,966
\\
5\,805\,523
\\
1\,786\,531
\\
438\,574
\\
83\,920
\\
11\,910
\\
1\,128
\\
54
\\
\end{tabular}\allowbreak
\begin{tabular}[t]{r}\sl j =\,2\\\hline
6\,220\,100
\\
32\,563\,392
\\
69\,906\,011
\\
86\,949\,004
\\
72\,306\,448
\\
43\,319\,066
\\
19\,471\,383
\\
6\,691\,382
\\
1\,760\,310
\\
348\,285
\\
49\,497
\\
4\,552
\\
204
\\
\end{tabular}\allowbreak
\begin{tabular}[t]{r}\sl j =\,3\\\hline
12\,943\,266
\\
50\,017\,890
\\
82\,616\,390
\\
80\,075\,156
\\
51\,859\,116
\\
23\,931\,700
\\
8\,099\,628
\\
2\,015\,912
\\
360\,113
\\
43\,247
\\
2\,982
\\
72
\\
\end{tabular}\allowbreak
\begin{tabular}[t]{r}\sl j =\,4\\\hline
17\,896\,018
\\
54\,312\,237
\\
71\,636\,757
\\
55\,598\,473
\\
28\,598\,913
\\
10\,261\,301
\\
2\,594\,341
\\
449\,940
\\
49\,420
\\
2\,802
\\
36
\\
\end{tabular}\allowbreak
\begin{tabular}[t]{r}\sl j =\,5\\\hline
18\,984\,001
\\
46\,813\,703
\\
50\,385\,897
\\
31\,674\,483
\\
12\,922\,375
\\
3\,532\,549
\\
634\,476
\\
68\,992
\\
3\,652
\\
36
\\
\end{tabular}\allowbreak
\begin{tabular}[t]{r}\sl j =\,6\\\hline
16\,719\,144
\\
34\,091\,282
\\
30\,117\,376
\\
15\,229\,199
\\
4\,813\,247
\\
955\,120
\\
110\,940
\\
6\,150
\\
72
\\
\end{tabular}\allowbreak
\begin{tabular}[t]{r}\sl j =\,7\\\hline
12\,772\,526
\\
21\,611\,154
\\
15\,535\,199
\\
6\,178\,619
\\
1\,450\,990
\\
193\,386
\\
12\,294
\\
204
\\
\end{tabular}\allowbreak
\begin{tabular}[t]{r}\sl j =\,8\\\hline
8\,656\,420
\\
12\,056\,388
\\
6\,911\,373
\\
2\,087\,824
\\
342\,608
\\
26\,967
\\
642
\\
\end{tabular}\allowbreak
\begin{tabular}[t]{r}\sl j =\,9\\\hline
5\,255\,640
\\
5\,923\,384
\\
2\,628\,658
\\
574\,674
\\
59\,983
\\
2\,161
\\
\end{tabular}\allowbreak
\begin{tabular}[t]{r}\sl j =\,10\\\hline
2\,865\,546
\\
2\,549\,764
\\
841\,681
\\
124\,393
\\
7\,011
\\
54
\\
\end{tabular}\allowbreak
\begin{tabular}[t]{r}\sl j =\,11\\\hline
1\,400\,474
\\
952\,258
\\
221\,559
\\
19\,969
\\
418
\\
\end{tabular}\allowbreak
\begin{tabular}[t]{r}\sl j =\,12\\\hline
610\,470
\\
303\,943
\\
46\,208
\\
2\,128
\\
\end{tabular}\allowbreak
\begin{tabular}[t]{r}\sl j =\,13\\\hline
235\,467
\\
81\,073
\\
7\,182
\\
114
\\
\end{tabular}\allowbreak
\begin{tabular}[t]{r}\sl j =\,14\\\hline
79\,458
\\
17\,461
\\
741
\\
\end{tabular}\allowbreak
\begin{tabular}[t]{r}\sl j =\,15\\\hline
23\,085
\\
2\,869
\\
38
\\
\end{tabular}\allowbreak
\begin{tabular}[t]{r}\sl j =\,16\\\hline
5\,643
\\
323
\\
\end{tabular}\allowbreak
\begin{tabular}[t]{r}\sl j =\,17\\\hline
1\,121
\\
19
\\
\end{tabular}\allowbreak
\begin{tabular}[t]{r}\sl j =\,18\\\hline
171
\\
\end{tabular}\allowbreak
\begin{tabular}[t]{r}\sl j =\,19\\\hline
18
\\
\end{tabular}\allowbreak
\begin{tabular}[t]{r}\sl j =\,20\\\hline
1
\\
\end{tabular}\allowbreak

}

\par\medskip
{\bf Book(``huck'', 23, 0, 0, 0, 1, 1, 0)}\ \par\nopagebreak
\begin{minipage}[t]{.2\textwidth}\vspace{0pt}
\includegraphics{graphs.12}
\end{minipage}
{\fontsize{5}{6}\selectfont \sf
\begin{tabular}[t]{r}\sl j =\,0\\\hline
0
\\
7\,644\,119\,040
\\
58\,063\,454\,208
\\
208\,089\,907\,200
\\
468\,472\,356\,864
\\
743\,860\,850\,688
\\
886\,362\,588\,672
\\
823\,088\,010\,752
\\
610\,456\,680\,992
\\
367\,568\,054\,960
\\
181\,618\,268\,880
\\
74\,123\,982\,824
\\
25\,065\,464\,810
\\
7\,022\,616\,847
\\
1\,625\,058\,718
\\
308\,561\,221
\\
47\,562\,773
\\
5\,855\,899
\\
562\,021
\\
40\,503
\\
2\,061
\\
66
\\
1
\\
\end{tabular}\allowbreak
\begin{tabular}[t]{r}\sl j =\,1\\\hline
7\,644\,119\,040
\\
130\,942\,365\,696
\\
714\,446\,189\,568
\\
2\,122\,537\,340\,160
\\
4\,123\,092\,911\,040
\\
5\,729\,638\,717\,440
\\
5\,999\,274\,355\,776
\\
4\,889\,305\,364\,240
\\
3\,167\,538\,981\,020
\\
1\,653\,497\,026\,704
\\
701\,112\,177\,494
\\
242\,350\,495\,650
\\
68\,266\,261\,385
\\
15\,600\,933\,383
\\
2\,866\,517\,167
\\
417\,203\,235
\\
46\,993\,465
\\
3\,949\,677
\\
233\,095
\\
8\,615
\\
150
\\
\end{tabular}\allowbreak
\begin{tabular}[t]{r}\sl j =\,2\\\hline
80\,523\,030\,528
\\
911\,365\,748\,352
\\
4\,072\,619\,596\,896
\\
10\,443\,363\,681\,456
\\
17\,837\,338\,328\,976
\\
21\,934\,086\,338\,856
\\
20\,332\,420\,289\,546
\\
14\,624\,737\,351\,937
\\
8\,313\,631\,620\,163
\\
3\,776\,758\,163\,442
\\
1\,378\,712\,232\,797
\\
404\,732\,343\,621
\\
95\,166\,437\,501
\\
17\,760\,717\,539
\\
2\,590\,259\,291
\\
288\,051\,462
\\
23\,498\,332
\\
1\,318\,791
\\
45\,170
\\
704
\\
\end{tabular}\allowbreak
\begin{tabular}[t]{r}\sl j =\,3\\\hline
427\,469\,042\,304
\\
3\,881\,430\,693\,504
\\
14\,940\,911\,766\,816
\\
33\,814\,295\,549\,328
\\
51\,446\,928\,857\,016
\\
56\,496\,623\,747\,504
\\
46\,713\,317\,299\,566
\\
29\,857\,553\,469\,301
\\
14\,994\,835\,055\,266
\\
5\,971\,670\,244\,733
\\
1\,892\,644\,074\,081
\\
476\,705\,833\,342
\\
94\,815\,567\,757
\\
14\,716\,870\,862
\\
1\,749\,939\,461
\\
155\,118\,437
\\
9\,852\,285
\\
422\,829
\\
11\,144
\\
144
\\
\end{tabular}\allowbreak
\begin{tabular}[t]{r}\sl j =\,4\\\hline
1\,534\,015\,292\,832
\\
11\,932\,392\,081\,456
\\
40\,626\,720\,391\,896
\\
82\,361\,720\,634\,568
\\
112\,779\,195\,195\,632
\\
111\,536\,656\,897\,631
\\
82\,898\,207\,721\,802
\\
47\,452\,347\,271\,542
\\
21\,231\,473\,215\,833
\\
7\,483\,614\,951\,821
\\
2\,082\,708\,405\,703
\\
456\,394\,199\,836
\\
78\,141\,581\,921
\\
10\,314\,463\,582
\\
1\,028\,275\,760
\\
75\,041\,770
\\
3\,809\,990
\\
122\,002
\\
1\,895
\\
\end{tabular}\allowbreak
\begin{tabular}[t]{r}\sl j =\,5\\\hline
4\,208\,425\,319\,232
\\
28\,958\,733\,158\,880
\\
88\,689\,338\,750\,268
\\
162\,872\,521\,275\,352
\\
202\,515\,265\,428\,917
\\
181\,822\,039\,012\,474
\\
122\,436\,175\,778\,571
\\
63\,289\,339\,398\,041
\\
25\,460\,955\,613\,478
\\
8\,026\,773\,100\,789
\\
1\,985\,823\,978\,447
\\
384\,178\,991\,077
\\
57\,622\,694\,859
\\
6\,604\,985\,046
\\
565\,949\,893
\\
35\,042\,180
\\
1\,483\,706
\\
38\,785
\\
485
\\
\end{tabular}\allowbreak
\begin{tabular}[t]{r}\sl j =\,6\\\hline
9\,456\,659\,968\,488
\\
58\,691\,487\,844\,000
\\
163\,670\,254\,561\,170
\\
274\,803\,257\,750\,421
\\
312\,764\,825\,917\,711
\\
256\,878\,631\,331\,512
\\
157\,948\,712\,034\,296
\\
74\,347\,451\,684\,553
\\
27\,140\,739\,158\,980
\\
7\,732\,248\,143\,990
\\
1\,720\,587\,723\,236
\\
297\,808\,086\,966
\\
39\,723\,622\,606
\\
4\,020\,794\,424
\\
301\,529\,428
\\
16\,128\,290
\\
575\,231
\\
11\,796
\\
86
\\
\end{tabular}\allowbreak
\begin{tabular}[t]{r}\sl j =\,7\\\hline
18\,198\,074\,885\,356
\\
103\,243\,788\,227\,460
\\
264\,694\,234\,344\,442
\\
409\,580\,434\,157\,718
\\
429\,835\,626\,654\,435
\\
325\,285\,245\,146\,851
\\
183\,994\,333\,180\,127
\\
79\,489\,180\,483\,326
\\
26\,556\,188\,103\,022
\\
6\,900\,191\,730\,740
\\
1\,394\,731\,774\,739
\\
218\,231\,222\,093
\\
26\,157\,204\,332
\\
2\,360\,083\,348
\\
155\,871\,340
\\
7\,184\,647
\\
210\,107
\\
3\,005
\\
\end{tabular}\allowbreak
\begin{tabular}[t]{r}\sl j =\,8\\\hline
30\,940\,731\,400\,736
\\
162\,102\,328\,114\,667
\\
385\,158\,972\,944\,741
\\
553\,143\,680\,550\,286
\\
538\,846\,162\,737\,722
\\
378\,230\,381\,624\,169
\\
198\,149\,002\,209\,146
\\
79\,122\,703\,627\,690
\\
24\,368\,441\,736\,674
\\
5\,818\,100\,553\,716
\\
1\,076\,252\,268\,769
\\
153\,327\,979\,453
\\
16\,622\,431\,041
\\
1\,344\,595\,660
\\
78\,646\,816
\\
3\,151\,857
\\
77\,523
\\
862
\\
\end{tabular}\allowbreak
\begin{tabular}[t]{r}\sl j =\,9\\\hline
47\,560\,610\,835\,303
\\
231\,997\,488\,723\,635
\\
514\,385\,945\,146\,530
\\
689\,925\,324\,841\,686
\\
627\,608\,139\,462\,721
\\
411\,029\,358\,220\,644
\\
200\,617\,849\,105\,619
\\
74\,482\,208\,530\,196
\\
21\,270\,851\,688\,207
\\
4\,692\,871\,669\,682
\\
798\,621\,154\,805
\\
104\,075\,037\,503
\\
10\,246\,508\,473
\\
745\,813\,853
\\
38\,768\,901
\\
1\,352\,207
\\
27\,469
\\
206
\\
\end{tabular}\allowbreak
\begin{tabular}[t]{r}\sl j =\,10\\\hline
67\,267\,401\,905\,395
\\
307\,596\,180\,791\,466
\\
640\,213\,989\,471\,637
\\
806\,372\,976\,235\,758
\\
688\,582\,006\,779\,197
\\
422\,898\,675\,284\,303
\\
193\,255\,210\,084\,387
\\
67\,024\,385\,760\,956
\\
17\,826\,935\,293\,467
\\
3\,648\,698\,696\,018
\\
573\,141\,852\,737
\\
68\,506\,947\,878
\\
6\,137\,792\,945
\\
402\,557\,913
\\
18\,595\,383
\\
561\,013
\\
9\,090
\\
33
\\
\end{tabular}\allowbreak
\begin{tabular}[t]{r}\sl j =\,11\\\hline
88\,747\,842\,449\,432
\\
382\,634\,254\,447\,952
\\
751\,445\,805\,849\,377
\\
893\,029\,971\,410\,258
\\
719\,045\,082\,827\,230
\\
415\,871\,654\,487\,327
\\
178\,629\,421\,459\,750
\\
58\,078\,504\,339\,287
\\
14\,431\,482\,593\,462
\\
2\,747\,096\,198\,020
\\
399\,057\,840\,816
\\
43\,803\,197\,182
\\
3\,573\,570\,819
\\
211\,208\,919
\\
8\,660\,607
\\
224\,483
\\
2\,783
\\
\end{tabular}\allowbreak
\begin{tabular}[t]{r}\sl j =\,12\\\hline
110\,413\,663\,746\,802
\\
451\,056\,497\,992\,889
\\
839\,500\,833\,190\,661
\\
945\,126\,777\,109\,455
\\
720\,198\,797\,592\,595
\\
393\,582\,445\,653\,266
\\
159\,375\,433\,931\,029
\\
48\,701\,630\,168\,789
\\
11\,328\,241\,057\,760
\\
2\,008\,371\,282\,725
\\
270\,019\,505\,090
\\
27\,225\,629\,630
\\
2\,022\,552\,038
\\
107\,776\,845
\\
3\,932\,249
\\
88\,116
\\
859
\\
\end{tabular}\allowbreak
\begin{tabular}[t]{r}\sl j =\,13\\\hline
130\,666\,491\,664\,749
\\
507\,874\,432\,372\,429
\\
899\,138\,156\,074\,953
\\
962\,134\,122\,445\,120
\\
695\,912\,053\,117\,113
\\
360\,278\,181\,979\,195
\\
137\,832\,995\,815\,031
\\
39\,652\,159\,130\,588
\\
8\,643\,948\,300\,111
\\
1\,428\,133\,964\,932
\\
177\,714\,498\,469
\\
16\,453\,312\,517
\\
1\,112\,678\,888
\\
53\,497\,484
\\
1\,739\,850
\\
33\,483
\\
234
\\
\end{tabular}\allowbreak
\begin{tabular}[t]{r}\sl j =\,14\\\hline
148\,114\,503\,232\,862
\\
549\,640\,583\,945\,429
\\
928\,434\,404\,732\,381
\\
946\,802\,956\,316\,230
\\
651\,526\,551\,686\,197
\\
320\,142\,255\,061\,978
\\
115\,885\,591\,893\,630
\\
31\,417\,553\,021\,990
\\
6\,421\,813\,488\,934
\\
988\,695\,612\,555
\\
113\,804\,068\,059
\\
9\,665\,564\,617
\\
594\,641\,211
\\
25\,811\,941
\\
748\,818
\\
12\,227
\\
53
\\
\end{tabular}\allowbreak
\begin{tabular}[t]{r}\sl j =\,15\\\hline
161\,710\,933\,451\,019
\\
574\,575\,879\,264\,822
\\
928\,302\,660\,980\,981
\\
904\,076\,789\,608\,510
\\
592\,898\,453\,733\,529
\\
276\,885\,385\,561\,807
\\
94\,920\,219\,709\,713
\\
24\,262\,330\,193\,502
\\
4\,649\,813\,570\,718
\\
666\,682\,299\,140
\\
70\,899\,825\,376
\\
5\,515\,962\,845
\\
308\,435\,706
\\
12\,095\,111
\\
312\,758
\\
4\,231
\\
8
\\
\end{tabular}\allowbreak
\begin{tabular}[t]{r}\sl j =\,16\\\hline
170\,813\,927\,289\,507
\\
582\,448\,562\,263\,319
\\
901\,802\,324\,461\,765
\\
840\,089\,439\,442\,061
\\
525\,705\,387\,527\,193
\\
233\,534\,388\,363\,802
\\
75\,855\,471\,078\,622
\\
18\,280\,905\,313\,399
\\
3\,283\,163\,455\,432
\\
437\,890\,395\,312
\\
42\,952\,342\,270
\\
3\,055\,400\,762
\\
155\,158\,798
\\
5\,504\,691
\\
126\,758
\\
1\,375
\\
\end{tabular}\allowbreak
\begin{tabular}[t]{r}\sl j =\,17\\\hline
175\,183\,508\,018\,224
\\
574\,307\,823\,880\,181
\\
853\,408\,768\,453\,265
\\
761\,337\,884\,261\,647
\\
454\,990\,478\,155\,838
\\
192\,359\,454\,489\,345
\\
59\,206\,435\,569\,975
\\
13\,447\,888\,149\,463
\\
2\,261\,157\,582\,200
\\
280\,092\,719\,250
\\
25\,286\,305\,141
\\
1\,641\,278\,936
\\
75\,679\,090
\\
2\,437\,347
\\
50\,082
\\
428
\\
\end{tabular}\allowbreak
\begin{tabular}[t]{r}\sl j =\,18\\\hline
174\,936\,105\,230\,968
\\
552\,153\,562\,397\,919
\\
788\,339\,037\,061\,129
\\
674\,053\,687\,878\,178
\\
384\,901\,878\,667\,513
\\
154\,895\,824\,102\,503
\\
45\,166\,410\,690\,866
\\
9\,661\,849\,230\,137
\\
1\,518\,926\,901\,479
\\
174\,390\,514\,342
\\
14\,453\,004\,774
\\
854\,165\,572
\\
35\,781\,667
\\
1\,050\,538
\\
19\,202
\\
121
\\
\end{tabular}\allowbreak
\begin{tabular}[t]{r}\sl j =\,19\\\hline
170\,475\,040\,838\,235
\\
518\,597\,408\,034\,174
\\
711\,981\,542\,551\,881
\\
583\,766\,784\,630\,227
\\
318\,587\,168\,246\,233
\\
122\,026\,125\,706\,545
\\
33\,691\,771\,782\,446
\\
6\,780\,631\,931\,834
\\
994\,941\,569\,885
\\
105\,615\,119\,324
\\
8\,011\,736\,241
\\
430\,162\,989
\\
16\,390\,015
\\
440\,428
\\
7\,108
\\
31
\\
\end{tabular}\allowbreak
\begin{tabular}[t]{r}\sl j =\,20\\\hline
162\,410\,989\,874\,779
\\
476\,549\,124\,971\,777
\\
629\,449\,797\,206\,205
\\
495\,042\,677\,879\,502
\\
258\,204\,924\,355\,942
\\
94\,096\,567\,995\,260
\\
24\,580\,817\,533\,069
\\
4\,647\,801\,211\,901
\\
635\,198\,727\,558
\\
62\,160\,275\,509
\\
4\,301\,315\,057
\\
209\,314\,543
\\
7\,263\,816
\\
178\,942
\\
2\,503
\\
7
\\
\end{tabular}\allowbreak
\begin{tabular}[t]{r}\sl j =\,21\\\hline
151\,482\,203\,761\,928
\\
428\,948\,242\,539\,842
\\
545\,264\,254\,521\,678
\\
411\,367\,624\,748\,301
\\
205\,018\,507\,541\,185
\\
71\,045\,733\,735\,891
\\
17\,540\,994\,280\,576
\\
3\,110\,824\,903\,986
\\
394\,989\,616\,510
\\
35\,513\,089\,086
\\
2\,232\,878\,866
\\
98\,214\,601
\\
3\,106\,323
\\
69\,836
\\
808
\\
1
\\
\end{tabular}\allowbreak
\begin{tabular}[t]{r}\sl j =\,22\\\hline
138\,481\,043\,916\,189
\\
378\,551\,418\,825\,180
\\
463\,156\,138\,345\,255
\\
335\,152\,391\,009\,976
\\
159\,540\,234\,730\,431
\\
52\,529\,965\,418\,036
\\
12\,241\,838\,274\,716
\\
2\,032\,190\,992\,977
\\
239\,034\,177\,695
\\
19\,668\,089\,937
\\
1\,118\,570\,285
\\
44\,326\,408
\\
1\,276\,909
\\
25\,927
\\
231
\\
\end{tabular}\allowbreak
\begin{tabular}[t]{r}\sl j =\,23\\\hline
124\,191\,059\,706\,563
\\
327\,779\,450\,801\,827
\\
385\,979\,736\,157\,458
\\
267\,823\,161\,239\,724
\\
121\,698\,403\,452\,623
\\
38\,034\,479\,409\,317
\\
8\,353\,299\,039\,147
\\
1\,294\,966\,296\,360
\\
140\,635\,158\,743
\\
10\,542\,647\,994
\\
539\,500\,322
\\
19\,181\,752
\\
502\,053
\\
9\,067
\\
56
\\
\end{tabular}\allowbreak
\begin{tabular}[t]{r}\sl j =\,24\\\hline
109\,337\,116\,734\,795
\\
278\,622\,686\,027\,317
\\
315\,714\,271\,443\,996
\\
209\,968\,589\,993\,354
\\
91\,005\,190\,493\,607
\\
26\,964\,153\,070\,620
\\
5\,570\,757\,950\,696
\\
804\,342\,907\,429
\\
80\,349\,078\,122
\\
5\,459\,753\,151
\\
249\,848\,063
\\
7\,926\,985
\\
187\,522
\\
2\,949
\\
10
\\
\end{tabular}\allowbreak
\begin{tabular}[t]{r}\sl j =\,25\\\hline
94\,549\,658\,362\,427
\\
232\,599\,415\,134\,395
\\
253\,533\,343\,091\,830
\\
161\,514\,900\,952\,720
\\
66\,709\,745\,029\,275
\\
18\,711\,968\,201\,815
\\
3\,629\,009\,765\,916
\\
486\,565\,736\,128
\\
44\,519\,508\,330
\\
2\,726\,196\,170
\\
110\,746\,672
\\
3\,111\,957
\\
65\,933
\\
883
\\
1
\\
\end{tabular}\allowbreak
\begin{tabular}[t]{r}\sl j =\,26\\\hline
80\,342\,944\,707\,533
\\
190\,758\,701\,968\,316
\\
199\,919\,236\,976\,320
\\
121\,905\,869\,974\,532
\\
47\,926\,985\,451\,590
\\
12\,706\,067\,491\,372
\\
2\,307\,827\,423\,945
\\
286\,375\,736\,913
\\
23\,887\,371\,320
\\
1\,309\,543\,395
\\
46\,804\,725
\\
1\,152\,202
\\
21\,528
\\
239
\\
\end{tabular}\allowbreak
\begin{tabular}[t]{r}\sl j =\,27\\\hline
67\,106\,067\,181\,304
\\
153\,717\,226\,791\,019
\\
154\,800\,992\,492\,735
\\
90\,270\,537\,601\,981
\\
33\,737\,964\,980\,454
\\
8\,438\,244\,908\,275
\\
1\,431\,651\,832\,106
\\
163\,814\,154\,069
\\
12\,391\,822\,214
\\
603\,613\,124
\\
18\,772\,999
\\
398\,423
\\
6\,401
\\
56
\\
\end{tabular}\allowbreak
\begin{tabular}[t]{r}\sl j =\,28\\\hline
55\,104\,744\,824\,358
\\
121\,719\,136\,468\,979
\\
117\,698\,380\,405\,855
\\
65\,567\,617\,176\,755
\\
23\,261\,782\,819\,227
\\
5\,477\,591\,996\,969
\\
865\,610\,892\,679
\\
90\,962\,163\,938
\\
6\,204\,115\,029
\\
266\,217\,054
\\
7\,105\,935
\\
127\,018
\\
1\,684
\\
10
\\
\end{tabular}\allowbreak
\begin{tabular}[t]{r}\sl j =\,29\\\hline
44\,491\,423\,252\,883
\\
94\,708\,476\,115\,968
\\
87\,858\,171\,960\,160
\\
46\,701\,038\,037\,329
\\
15\,701\,786\,388\,350
\\
3\,473\,196\,092\,061
\\
509\,618\,740\,844
\\
48\,964\,445\,869
\\
2\,991\,921\,656
\\
111\,984\,386
\\
2\,521\,253
\\
36\,746
\\
380
\\
1
\\
\end{tabular}\allowbreak
\begin{tabular}[t]{r}\sl j =\,30\\\hline
35\,321\,028\,354\,267
\\
72\,405\,228\,468\,100
\\
64\,373\,528\,302\,495
\\
32\,605\,441\,633\,781
\\
10\,370\,366\,228\,563
\\
2\,149\,485\,588\,138
\\
291\,838\,357\,248
\\
25\,513\,055\,441
\\
1\,386\,738\,494
\\
44\,763\,707
\\
831\,602
\\
9\,466
\\
72
\\
\end{tabular}\allowbreak
\begin{tabular}[t]{r}\sl j =\,31\\\hline
27\,569\,839\,915\,839
\\
54\,377\,952\,292\,000
\\
46\,281\,448\,067\,004
\\
22\,303\,581\,464\,318
\\
6\,697\,239\,653\,980
\\
1\,297\,255\,575\,747
\\
162\,366\,310\,290
\\
12\,846\,282\,872
\\
616\,213\,381
\\
16\,929\,544
\\
252\,202
\\
2\,112
\\
11
\\
\end{tabular}\allowbreak
\begin{tabular}[t]{r}\sl j =\,32\\\hline
21\,155\,277\,285\,691
\\
40\,108\,145\,965\,420
\\
32\,636\,626\,427\,969
\\
14\,939\,530\,821\,672
\\
4\,226\,027\,636\,156
\\
762\,739\,720\,158
\\
87\,644\,734\,347
\\
6\,238\,791\,356
\\
261\,761\,786
\\
6\,025\,245
\\
69\,220
\\
386
\\
1
\\
\end{tabular}\allowbreak
\begin{tabular}[t]{r}\sl j =\,33\\\hline
15\,954\,847\,768\,088
\\
29\,043\,492\,273\,524
\\
22\,562\,646\,009\,720
\\
9\,792\,558\,415\,254
\\
2\,603\,398\,119\,525
\\
436\,426\,275\,281
\\
45\,833\,115\,884
\\
2\,915\,994\,723
\\
105\,936\,834
\\
2\,004\,377
\\
16\,805
\\
52
\\
\end{tabular}\allowbreak
\begin{tabular}[t]{r}\sl j =\,34\\\hline
11\,823\,018\,778\,121
\\
20\,638\,880\,703\,615
\\
15\,283\,132\,698\,005
\\
6\,276\,719\,028\,345
\\
1\,564\,280\,125\,341
\\
242\,715\,583\,856
\\
23\,180\,300\,406
\\
1\,308\,416\,720
\\
40\,682\,440
\\
617\,953
\\
3\,491
\\
4
\\
\end{tabular}\allowbreak
\begin{tabular}[t]{r}\sl j =\,35\\\hline
8\,605\,270\,288\,428
\\
14\,385\,457\,021\,771
\\
10\,136\,458\,047\,648
\\
3\,930\,854\,880\,338
\\
915\,803\,054\,713
\\
131\,021\,687\,166
\\
11\,316\,367\,286
\\
561\,964\,518
\\
14\,753\,318
\\
174\,625
\\
589
\\
\end{tabular}\allowbreak
\begin{tabular}[t]{r}\sl j =\,36\\\hline
6\,149\,017\,377\,510
\\
9\,828\,892\,336\,538
\\
6\,577\,897\,351\,181
\\
2\,403\,021\,466\,599
\\
521\,796\,689\,085
\\
68\,546\,193\,842
\\
5\,320\,896\,693
\\
230\,239\,072
\\
5\,022\,558
\\
44\,548
\\
73
\\
\end{tabular}\allowbreak
\begin{tabular}[t]{r}\sl j =\,37\\\hline
4\,311\,436\,549\,811
\\
6\,578\,623\,394\,399
\\
4\,173\,019\,367\,654
\\
1\,432\,512\,031\,399
\\
288\,971\,661\,588
\\
34\,695\,422\,122
\\
2\,403\,511\,898
\\
89\,612\,569
\\
1\,593\,264
\\
10\,031
\\
5
\\
\end{tabular}\allowbreak
\begin{tabular}[t]{r}\sl j =\,38\\\hline
2\,964\,474\,698\,957
\\
4\,310\,052\,558\,697
\\
2\,585\,645\,699\,396
\\
831\,790\,371\,315
\\
155\,326\,381\,067
\\
16\,957\,718\,111
\\
1\,039\,924\,955
\\
32\,969\,506
\\
466\,505
\\
1\,925
\\
\end{tabular}\allowbreak
\begin{tabular}[t]{r}\sl j =\,39\\\hline
1\,997\,470\,496\,058
\\
2\,761\,684\,207\,395
\\
1\,563\,111\,427\,998
\\
469\,846\,262\,712
\\
80\,906\,285\,732
\\
7\,985\,680\,153
\\
429\,469\,516
\\
11\,395\,789
\\
124\,530
\\
297
\\
\end{tabular}\allowbreak
\begin{tabular}[t]{r}\sl j =\,40\\\hline
1\,317\,887\,697\,673
\\
1\,728\,990\,256\,047
\\
920\,894\,742\,411
\\
257\,819\,091\,781
\\
40\,765\,822\,575
\\
3\,614\,180\,128
\\
168\,586\,789
\\
3\,672\,196
\\
29\,815
\\
33
\\
\end{tabular}\allowbreak
\begin{tabular}[t]{r}\sl j =\,41\\\hline
850\,665\,829\,404
\\
1\,056\,510\,698\,738
\\
528\,044\,208\,920
\\
137\,219\,085\,706
\\
19\,829\,960\,250
\\
1\,567\,456\,613
\\
62\,586\,269
\\
1\,092\,398
\\
6\,261
\\
2
\\
\end{tabular}\allowbreak
\begin{tabular}[t]{r}\sl j =\,42\\\hline
536\,654\,329\,284
\\
629\,363\,523\,349
\\
294\,272\,153\,267
\\
70\,712\,799\,890
\\
9\,291\,403\,735
\\
649\,222\,792
\\
21\,837\,068
\\
296\,134
\\
1\,117
\\
\end{tabular}\allowbreak
\begin{tabular}[t]{r}\sl j =\,43\\\hline
330\,528\,626\,041
\\
365\,006\,585\,603
\\
159\,130\,489\,574
\\
35\,214\,172\,506
\\
4\,182\,750\,513
\\
255\,780\,213
\\
7\,105\,191
\\
71\,881
\\
161
\\
\end{tabular}\allowbreak
\begin{tabular}[t]{r}\sl j =\,44\\\hline
198\,505\,264\,692
\\
205\,792\,941\,637
\\
83\,350\,267\,190
\\
16\,908\,923\,264
\\
1\,803\,797\,937
\\
95\,399\,665
\\
2\,134\,376
\\
15\,238
\\
17
\\
\end{tabular}\allowbreak
\begin{tabular}[t]{r}\sl j =\,45\\\hline
116\,089\,433\,466
\\
112\,609\,389\,066
\\
42\,202\,405\,360
\\
7\,809\,196\,055
\\
742\,649\,893
\\
33\,492\,231
\\
584\,182
\\
2\,717
\\
1
\\
\end{tabular}\allowbreak
\begin{tabular}[t]{r}\sl j =\,46\\\hline
66\,010\,296\,070
\\
59\,693\,833\,147
\\
20\,608\,991\,452
\\
3\,458\,946\,792
\\
290\,757\,073
\\
10\,990\,747
\\
143\,089
\\
383
\\
\end{tabular}\allowbreak
\begin{tabular}[t]{r}\sl j =\,47\\\hline
36\,432\,673\,762
\\
30\,590\,692\,119
\\
9\,681\,466\,532
\\
1\,464\,504\,260
\\
107\,745\,798
\\
3\,342\,402
\\
30\,575
\\
38
\\
\end{tabular}\allowbreak
\begin{tabular}[t]{r}\sl j =\,48\\\hline
19\,480\,581\,012
\\
15\,119\,155\,354
\\
4\,362\,140\,777
\\
590\,430\,015
\\
37\,582\,223
\\
931\,895
\\
5\,485
\\
2
\\
\end{tabular}\allowbreak
\begin{tabular}[t]{r}\sl j =\,49\\\hline
10\,069\,527\,361
\\
7\,187\,420\,893
\\
1\,878\,612\,930
\\
225\,629\,513
\\
12\,256\,579
\\
234\,974
\\
776
\\
\end{tabular}\allowbreak
\begin{tabular}[t]{r}\sl j =\,50\\\hline
5\,019\,400\,819
\\
3\,276\,272\,681
\\
770\,201\,912
\\
81\,283\,454
\\
3\,706\,816
\\
52\,641
\\
77
\\
\end{tabular}\allowbreak
\begin{tabular}[t]{r}\sl j =\,51\\\hline
2\,406\,133\,696
\\
1\,426\,882\,777
\\
299\,178\,401
\\
27\,422\,357
\\
1\,029\,033
\\
10\,232
\\
4
\\
\end{tabular}\allowbreak
\begin{tabular}[t]{r}\sl j =\,52\\\hline
1\,105\,659\,524
\\
591\,250\,737
\\
109\,478\,461
\\
8\,592\,757
\\
258\,783
\\
1\,668
\\
\end{tabular}\allowbreak
\begin{tabular}[t]{r}\sl j =\,53\\\hline
485\,226\,772
\\
231\,936\,330
\\
37\,477\,493
\\
2\,474\,897
\\
57\,926
\\
216
\\
\end{tabular}\allowbreak
\begin{tabular}[t]{r}\sl j =\,54\\\hline
202\,490\,441
\\
85\,622\,599
\\
11\,898\,600
\\
646\,355
\\
11\,258
\\
20
\\
\end{tabular}\allowbreak
\begin{tabular}[t]{r}\sl j =\,55\\\hline
79\,941\,549
\\
29\,531\,012
\\
3\,465\,267
\\
150\,277
\\
1\,830
\\
1
\\
\end{tabular}\allowbreak
\begin{tabular}[t]{r}\sl j =\,56\\\hline
29\,674\,484
\\
9\,430\,560
\\
912\,617
\\
30\,308
\\
234
\\
\end{tabular}\allowbreak
\begin{tabular}[t]{r}\sl j =\,57\\\hline
10\,280\,210
\\
2\,757\,016
\\
213\,215
\\
5\,101
\\
21
\\
\end{tabular}\allowbreak
\begin{tabular}[t]{r}\sl j =\,58\\\hline
3\,293\,359
\\
727\,123
\\
43\,020
\\
673
\\
1
\\
\end{tabular}\allowbreak
\begin{tabular}[t]{r}\sl j =\,59\\\hline
964\,438
\\
169\,643
\\
7\,205
\\
62
\\
\end{tabular}\allowbreak
\begin{tabular}[t]{r}\sl j =\,60\\\hline
254\,353
\\
34\,073
\\
940
\\
3
\\
\end{tabular}\allowbreak
\begin{tabular}[t]{r}\sl j =\,61\\\hline
59\,228
\\
5\,661
\\
85
\\
\end{tabular}\allowbreak
\begin{tabular}[t]{r}\sl j =\,62\\\hline
11\,848
\\
730
\\
4
\\
\end{tabular}\allowbreak
\begin{tabular}[t]{r}\sl j =\,63\\\hline
1\,956
\\
65
\\
\end{tabular}\allowbreak
\begin{tabular}[t]{r}\sl j =\,64\\\hline
250
\\
3
\\
\end{tabular}\allowbreak
\begin{tabular}[t]{r}\sl j =\,65\\\hline
22
\\
\end{tabular}\allowbreak
\begin{tabular}[t]{r}\sl j =\,66\\\hline
1
\\
\end{tabular}\allowbreak

}



\makeatletter
\renewcommand{\@setaddresses}{}
\makeatother


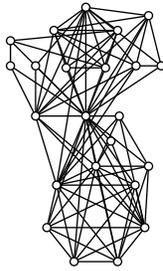
\end{document}